\documentclass[9pt]{amsart}
\usepackage[foot]{amsaddr}
\usepackage{geometry}
\usepackage{amsmath}%
\usepackage{amsfonts}%
\usepackage{amssymb}%
\usepackage{graphicx}

\usepackage{xspace}
\usepackage{caption}
\usepackage[position=top]{subcaption}
\usepackage{mathtools}
\usepackage{multirow}
\usepackage{nicefrac}
\usepackage{braket}
\usepackage{tabularx}
\usepackage{scalerel,stackengine}
\usepackage{hyperref}

\usepackage{xcolor}

\newtheorem{theorem}{Theorem}
\theoremstyle{plain}

\newtheorem{corollary}{Corollary}

\newtheorem{lemma}{Lemma}

\newtheorem{proposition}{Proposition}

\newtheorem{question}{Question}

\numberwithin{equation}{section}
\begin{document}
\title[Harmonic sandpile dynamics]{Harmonic dynamics of the Abelian sandpile}
\author{Moritz Lang}
\email[M. Lang]{moritz.lang@ist.ac.at}
\author{Mikhail Shkolnikov}
\email[M.~Shkolnikov]{mikhail.shkolnikov@ist.ac.at}
\address[M. Lang and M. Shkolnikov]
{Institute of Science and Technology Austria, 
Am Campus 1, 3400 Klosterneuburg, Austria}%
\thanks{M.~Lang is grateful to the members of the Guet and Tka\v{c}ik groups of IST Austria for valuable discussions, comments and support.\hfil\break\indent M.~Shkolnikov is grateful to Nikita Kalinin for the inspiring communications.}
\date{\today}
\subjclass{60K35, 47D07, 20K01} %
\keywords{Abelian sandpile, identity, dynamics, criticality}%

\begin{abstract}
The abelian sandpile serves as a simple model system to study self-organized criticality, a phenomenon occurring in many important biological, physical and social processes. The identity element of the abelian group is a fractal composed of self-similar patches with periodic patterns, and the limit of this identity is subject of extensive collaborative research. Here, we analyze the evolution of the sandpile identity under harmonic fields of different orders. We show that this evolution corresponds to periodic cycles through the abelian group characterized -- to a large extend -- by the smooth transformation and apparent conservation of the patches constituting the identity. The dynamics induced by second and third order harmonic fields resemble smooth stretchings, respectively translations, of the identity, while the ones induced by fourth order harmonic fields resemble magnifications and rotations. Starting with order three, the dynamics pass through extended regions of seemingly random configurations which spontaneously reassemble into accentuated patterns. We show that the space of harmonic functions projects to the extended analogue of the sandpile group, thus providing a set of universal coordinates identifying configurations between different domains. Since the original sandpile group is a subgroup of the extended one, this directly implies that the sandpile group admits a natural renormalization. 
Furthermore, we show that the harmonic fields can be induced by simple Markov processes, and that the corresponding stochastic sandpile identity dynamics show remarkable robustness over dozens or hundreds of periods. 
Finally, we encode information into seemingly random sandpile configurations, and decode this information with an algorithm requiring minimal prior knowledge. Our results suggest that harmonic fields might split the sandpile group into sub-sets showing different critical coefficients, and that it might be possible to extend the fractal structure of the sandpile identity beyond the boundaries of its finite domain.
\end{abstract}
\maketitle

\section{Introduction}
\noindent Self-organized criticality (SOC) is the property of dissipative systems driven by fluctuating forces to automatically converge into critical configurations which eventually become unstable and relax in processes referred to as avalanches, characterized by scale-free spatio-temporal correlations \cite{Bak1987,Aschwanden2016}. Different to critical points e.g. of systems known from equilibrium statistical mechanics, SOC thus does not require the tuning of parameters, like the temperature or pressure, to reach criticality \cite[p.~2]{Bak1987,Paoletti2014}. The first model showing SOC was introduced by Bak, Tang and Wiesenfeld in 1987 \cite{Bak1987}. Nowadays known as the BTW model, it describes the evolution of an idealized sandpile under random dropping (addition) of grains of sand \cite{Bak1987}. The concept of SOC subsequently spread into various scientific fields \cite[p.~3ff]{Paoletti2014}, and was utilized to explain phenomena as diverse as solar flares \cite{Lu1991}, earthquakes \cite{Olami1992}, forest fires \cite{Malamud1998}, natural evolution \cite{Sneppen1995}, neuron firing \cite{Levina2007}, stock market crashes \cite[p.~317ff]{Sornette2017}, and similar (see \cite{Aschwanden2016} for a recent review).

While this success story resulted in a broad range of models showing SOC available today, the original BTW model remains an active field of research due its simplicity and its intriguing mathematical properties, rendering it the archetypical model for SOC \cite[p.~1]{Paoletti2014}. The BTW sandpile model is a cellular automaton defined on a rectangular $N\times M$ domain $\Gamma\subset\mathbb{Z}^2$ of the standard square lattice $\mathbb{Z}^2$ (Figure~\ref{Fig1}A) \cite{Bak1987}. Each vertex $(i,j)$ of the domain carries a non-negative number $c_{i,j}$ of particles (``grains of sand''), with $C\in\mathbb{N}_{\geq 0}^{N\times M}$ referred to as the configuration of the sandpile. Starting from some initial configuration, particles are slowly dropped onto vertices chosen at random. When during this process the number $c_{i,j}$ of particles of any vertex exceeds three, this vertex becomes unstable and ``topples'', decreasing the number of its particles by four and increasing the number of particles of each of its direct neighbors by one. Thus, toppling of vertices in the interior of the domain conserves the total number $|C|$ of particles in the sandpile, whereas toppling of vertices at the sides and the corners of the domain decreases the total number by one and two, respectively. The redistribution of particles due to the toppling of a vertex can render other vertices unstable, resulting in subsequent topplings in a process referred to an ``avalanche''. Due to the loss of particles at the boundaries of the domain, this process eventually ends \cite[Theorem~1]{Creutz1990}, and the ``relaxed'' sandpile reaches a stable configuration which is independent of the specific order in which the vertices were toppled \cite[p.~13]{Paoletti2014}. The distribution of avalanche sizes -- the total number of topplings after a random particle drop -- follows a power law \cite{Bak1987} and is thus scale-invariant. However, the critical exponent for this power law is yet unknown \cite{Bhupatiraju2017}.

\begin{figure}[htb]
	\centering
	\includegraphics[width=0.67\linewidth]{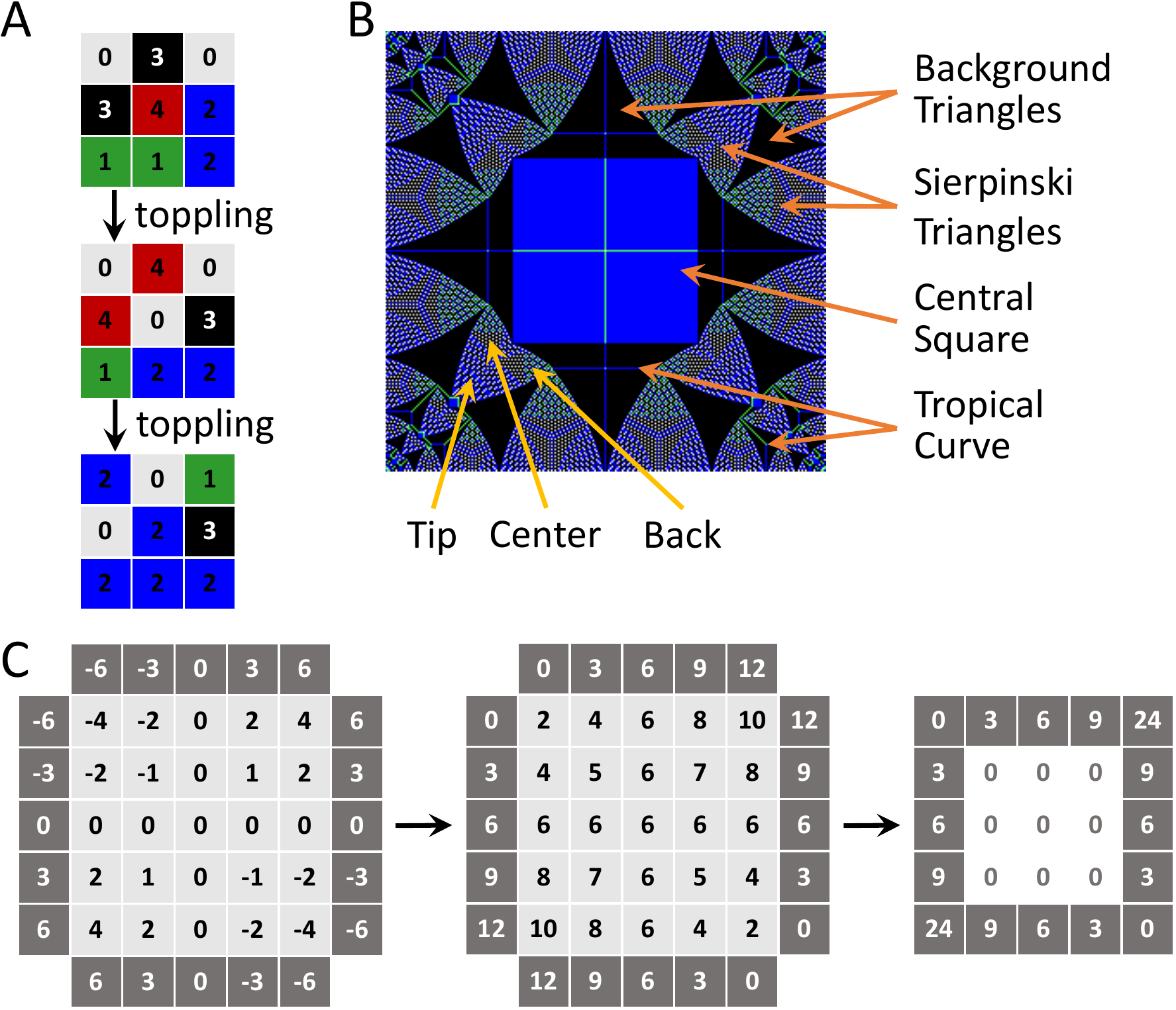}
	\caption{The BTW sandpile model, the sandpile identity and the construction of potentials. 
	A) Each vertex of the $N\times M$ domain (here: $3\times 3$) carries a non-negative number of particles, and additional particles are dropped onto randomly chosen vertices. If during this process a vertex happens to carry four or more particles, it becomes unstable and topples, decreasing the number of its particles by four and increasing the number of particles carried by each of its (four or less) neighbors by one. Thus, the toppling of a vertex in the interior of the domain conserves the total number of particles in the sandpile, while the toppling of a vertex at the edge or the corner of the domain decreases the total number of particles by one, respectively two. The toppling of one vertex can render other, previously stable vertices unstable, resulting in an ``avalanche'' of subsequent topplings. Every avalanche eventually stops, resulting in a stable configuration which is independent from the order in which vertices were toppled \cite{Creutz1990}. 
	B) The sandpile identity on a $255\times 255$ square domain. White pixels represent vertices carrying zero particles, while green, blue and black pixels represent vertices carrying one, two or three particles, respectively. We distinguish between three different two-dimensional patches in the sandpile identity (orange arrows): (i) The central square consists of vertices carrying two particles each; (ii) background triangles consist of vertices carrying three particles each; and (iii) Sierpinski triangles \cite[p.~109ff]{Paoletti2014}. The latter are composed of three different patterns (yellow arrows), to which we refer to as the tips, centers and backs of the Sierpinski triangles. Additionally, thin one-dimensional ``curves'' or ``strings'', referred to as tropical curves, occur in the sandpile identity, slightly disturbing the patches they cross.
	C) To construct the potential for a harmonic (here: $h_{ij}^{2a}= i j$) on a given $N\times M$ rectangular (here: $5\times 5$), the domain is first extended by adding an extra vertex to each outward facing side of a boundary vertex. To each vertex of the extended domain, we then associate the corresponding (possibly negative) value of the harmonic function at this position, with the origin of the harmonic function assumed to lie at the center of the domain for both $N$ and $M$ odd, or at a vertex close to the center otherwise. Subsequently, we subtract the minimal value of the harmonic function restricted to the extended domain (here: $-6$) from the values associated to each vertex. If, after this step, the greatest common divisor of all values is greater than one, we divide by this number. Finally, the potential is constructed by ``folding'' the vertices with which the original domain was extended into a new empty domain. At the corners where two ``folded vertices'' overlap, the sum of their values is taken to determine the corresponding value of the potential.
	}
	\label{Fig1}
\end{figure}

Besides SOC, the sandpile model possesses several mathematical properties not only simplifying its analysis, but being of specific interest themselves.
The set of recurrent configurations of the BTW model -- all stable configurations which can be reached from any other configuration by dropping particles -- form an abelian group \cite{Dhar1995}. The identity of this group, the sandpile or Creutz identity -- after Michael Creutz who first studied it in depth \cite{Creutz1990} -- shows a remarkably complex self-similar fractal structure composed of patches covered with periodic patterns (``textures'', see Figure~\ref{Fig1}B) which is still not completely understood. 
For some configurations different to the sandpile identity, scaling limits for infinite domains were shown to exist, and the patches visible in these configurations as well as their robustness was analyzed \cite{Pegden2013,Levine2016,Pegden2017}. Corresponding results for the sandpile identity -- like a closed formula for its construction -- are still missing \cite[p.~61]{Paoletti2014}, even though recently a proof for the scaling limit of the sandpile identity was announced \cite{Sportiello2015}. At the time of writing, however, only rigorous proofs for some specific structural aspects of the sandpile identity are available \cite{LeBorgne2002}.
For example, the thin (usually one pixel wide) curves or ``strings'' visible in the identity (Figure~\ref{Fig1}B) were recently identified as tropical curves \cite{Caracciolo2010,Kalinin2015,Kalinin2016,Kalinin2017}, structures from tropical geometry arising e.g. in some areas of string theory and statistical physics.

In this article we study a -- to our knowledge -- yet unknown property of the sandpile model, namely the evolution of the sandpile identity under harmonic fields externally imposed by deterministically or stochastically dropping particles on boundary vertices of the domain.
We show that such harmonic fields induce cyclic dynamics of the sandpile identity through the abelian group, smoothly transforming individual patches and tropical curves, mapping them onto one another or merging them into different objects. Specifically, we show that first order harmonic fields smoothly map tropical curves onto one another. Second order harmonic fields resemble stretching actions, map the self-similar patches constituting the sandpile identity onto one another, and form new two-dimensional patches out of many one-dimensional tropical curves. Third order harmonic fields resemble translational actions, and display the emergence of several new fractal structures formed out of patches with novel periodic patterns. Finally, fourth order harmonics resemble zooming actions, display the emergence of many regularly spaced fractal structures each restricted to some region of the domain, eventually forming what we will refer to as hyper-fractals. We introduce an extended analogue of the sandpile model where each vertex at the domain boundary is allowed to carry a real number of particles. For this extended model, we show that harmonic fields define closed geodesics and thus provide a set of universal coordinates for sandpile groups corresponding to different domains. Since there exist a natural inclusion of the original sandpile group into the extended one and a floor-projection from the extended group to the original one, this directly implies that the sandpile group admits a natural renormalization.
We subsequently extend our theory to super-harmonic fields, and show that it is possible to induce harmonic and super-harmonic fields by simple Markov processes. The stochastic dynamics resulting from the latter show remarkable robustness, allow to encode information into seemingly random configurations, and to decode this information with a simple and memoryless process.

\section{Results}
\subsection{Motivation}
In this article, we study the evolution of the sandpile identity under integer valued harmonic fields $H$. To motivate our study, recall that the stable configuration $C=(C^u)^\circ$ reached when relaxing an unstable configuration $C^u$, with $(\cdot)^\circ$ the relaxation operator corresponding to a series of topplings resulting in a stable sandpile, can be expressed in terms of the toppling function $H$, where $h_{i,j}$ quantifies how often the vertex $(i,j)$ toppled while relaxing the sandpile from $C^u$ to $C$ \cite{Fey2010}:
\begin{equation}\label{topplingEq}
C=(C^u)^\circ=C^u+\Delta H,
\end{equation}
with $\Delta$ the discrete Laplace operator defined by
\begin{equation*}
(\Delta H)_{i,j}=h_{i+1,j}+h_{i-1,j}+h_{i,j+1}+h_{i,j-1}-4h_{i,j}.
\end{equation*}
Note that we adopt the convention to set $h_{i,j}=0$ for any $(i,j)$ outside of the domain. 

In general, it is non-trivial to find the toppling function $H$ corresponding to the relaxation of a given unstable configuration $C^u$ without performing the relaxation itself. However, assume that it is possible to construct $C^u$ by adding $x_{ij}$ particles to each vertex $(i,j)$ of some ``initial'' recurrent configuration $C^0$ such that $H$ becomes harmonic in the interior of the domain, i.e. such that $(\Delta H)_{ij}=0$ for all vertices $(i,j)$ in the interior. Then, it follows from Eq.~\ref{topplingEq} that the relaxation of $C^u$ only changes the particle numbers of vertices at the boundary of the domain, and that $c_{i,j}=c^u_{i,j}$ for all vertices $(i,j)$ in the interior. Now, instead of adding all particles in $X$ at once resulting in general in a big avalanche, we could equivalently add one particle after the other, relaxing the sandpile every time a vertex becomes unstable. Since the addition of particles and the toppling operator commute \cite{Creutz1990}, we would still arrive at the same final configuration $C$, but during this process we would pass several intermediate, stable configurations. How do these intermediate configurations look like, and how should we choose $C^0$, $X$ and the order in which we add the particles $x_{ij}$?

Given these questions, the idea of this article can be easily explained: We start at an initial configuration $C^0=I$ corresponding to the sandpile identity $I$, we add $x_{ij}=(-\Delta H)_{ij}$ particles to each vertex, and we add these particles such that particle additions to the same vertex are uniformly spaced in time.
The motivation behind the first two ingredients of our algorithm is as follows: Recall that $H$ should be harmonic in the interior of the domain, which directly implies that we only add particles to the boundaries of the domain, i.e. $x_{ij}=(-\Delta H)_{ij}=0$ for every vertex $(i,j)$ in the interior. Then, when setting $C^u=C^0+X$, $C^0=I$ and $X=-\Delta H$ in Eq.~\ref{topplingEq}, we obtain
\begin{equation*}
(I+X)^\circ=I-\Delta H+\Delta H=I.
\end{equation*}
In other words, by adding only particles to the boundary of the domain, we arrive again at the sandpile identity after we added \textit{all} particles. However, since we add particles one by one and relax the sandpile after each step, we will see several intermediate, stable configurations during this process. Intuitively, these intermediate configurations will first become more and more dissimilar from the sandpile identity, before finally converging back to it. How these ``oscillatory dynamics'' look like in detail is what we analyze in this article.

The procedure described above is reminiscent of a method proposed nearly $30$ years ago by Michael Creutz to construct the sandpile identity stating from the configuration $c^0_{ij}=0$ where each vertex carries zero particles. \cite{Creutz1990}. In this method, a single particle is added to each vertex at the sides of the rectangular domain, two particles to each vertex at a corner, and the sandpile is relaxed \cite{Creutz1990}. These steps are repeated several times, and it turns out that the sandpile will eventually convergence to the identity \cite{Creutz1990}. Notably, this corresponds to adding in each step $X=-\Delta H^0$ to the configuration, with $h^0_{ij}=1$ the constant harmonic function. However, as it will become clear in the following, the constant harmonic function $H^0$ will not result in any non-trivial dynamics, and -- to our knowledge -- nobody yet analyzed the effect of adding particles $X$ according to higher order harmonics functions, and specifically not to add these particles one-by-one instead of all-at-once.

\subsection{Harmonic dynamics of the sandpile identity}
Let $H$ be a non-negative integer valued harmonics function on the lattice points in the interior of an $N\times M$ rectangular domain, taking non-positive values of the Laplacian at the boundaries, i.e. $h_{ij}\geq 0$ and $(\Delta H)_{ij}\leq 0$ for all vertices $(i,j)$, and $(\Delta H)_{ij}=0$ for all vertices in the interior. Furthermore, let $X^H=-\Delta H$ be the negative of the discrete Laplacian of $H$, to which we refer as the potential of the harmonic function in the following. Note, that this implies that $x^H_{ij}=0$ for $(i,j)$ in the interior of the domain, and $x^H_{ij}\geq 0$ at the boundary.

We then define the dynamics $I^H(t)$ of the sandpile identity $I$ at time $t\geq 0$ induced by the harmonic field $H$ by
\begin{align}\label{dynDef}
I^H(t)=(I+\lfloor tX^H\rfloor)^\circ,
\end{align}
with $\lfloor .\rfloor$ the floor function. We refer to the Discussion section for possible modifications of this algorithm, e.g. for using the ceil or round function instead of the floor function.

Since the Laplace operator is linear, the weighted sum $k^1H^1+k^2H^2$, with $k^1,k^2\in\mathbb{N}_{\geq 0}$ of two harmonic functions satisfying the conditions above is harmonic, too, and also satisfies the other conditions. Due to the non-negativity constraints, the set of all $H$ however does not form a vector space. Furthermore, we note that, in general, $\lfloor tX^{H^1}\rfloor+\lfloor tX^{H^2}\rfloor$ is different from $\lfloor tX^{H^1}+tX^{H^2}\rfloor$. Despite this non-linearity, we show at the end of this section that the action induced by the sum of two harmonic fields is well approximated by the sum of the actions of the individual harmonic fields.

To restrict our analysis to a reasonably small and computationally feasible set of harmonic functions fulfilling the requirements above, we constructed nine harmonic functions forming the basis of the vector space of all (possibly negative valued) harmonic functions on the whole lattice of order four or less. Specifically, there exists one of such basis vectors for order zero harmonics, and two additional basis vectors for each additional order. The set of all these basis vectors is listed in Table~\ref{table1}. The algorithm to construct from this basis non-negative harmonic functions for a given $N\times M$ rectangular domain fulfilling the requirements above is depicted in Figure~\ref{Fig1}C. In short, this algorithm places the origin of the harmonic function onto the center of the rectangle, and then keeps adding the constant harmonic function $h^0_{ij}=1$ until $h_{ij}+kh^0_{ij}\geq 0$ and $x_{ij}=-(\Delta (H+kH^0))_{ij}\geq 0$ for all vertices $(i,j)$. 
Since, due to this construction, the harmonics corresponding to the same basis vectors only differ by a constant value for different domain sizes, we will refer in the following to the harmonic and its corresponding basis vector interchangeably.

\begin{table}
\begin{tabularx}{0.99\linewidth}{lllX}
		\textbf{Order} 				& \textbf{Identifier} & \textbf{Harmonic} & \textbf{Action} \\
		\hline
		$0$ 									& $H^{0}$ 	& $1$																						& None\\
		\multirow{2}{*}{1} 		& $H^{1a}$ 	& $i$																						& Hor. movement of tropical curves\\
													& $H^{1b}$	& $j$																						& Vert. movement of tropical curves\\
		\multirow{2}{*}{2} 		& $H^{2a}$	& $i j$																					& Diagonal stretching\\
													& $H^{2b}$	& $i^2-j^2$																			& Vert.\&hor. stretching\\
		\multirow{2}{*}{3} 		& $H^{3a}$	& $i^3-3i j^2$																	& Horizontal translation + stretching\\
													& $H^{3b}$	& $j^3-3j i^2$																	& Vertical translation + stretching\\
		\multirow{2}{*}{4} 		& $H^{4a}$	& $i^4-6 i^2 j^2 + j^4 - i^2 - j^2$							& Zooming\\
													& $H^{4b}$	& $i^3 j-i j^3$																	& Zooming + rotation\\
		\hline\\
	\end{tabularx}
\caption{Basis vectors of the set of discrete harmonic functions of order four or less used in this study. The column ``action'' provides an intuitive interpretation of the corresponding sandpile identity dynamics $I^H(t)$. Due to the complexity of the dynamics, this interpretation is however not in a one-to-one correspondence with common definitions of the respective action. Note that the discrete Laplacian leads to a slightly different set of basis vectors than the continuous one. For example, $\Delta_c(y^4-6 y^2 x^2 + x^4-x^2-y^2)=-4$ while $\Delta(i^4-6 i^2 j^2 + j^4- i^2 - j^2)=0$.}
\label{table1}
\end{table}

\begin{figure}[htb]
	\centering
	\includegraphics[width=0.73\linewidth]{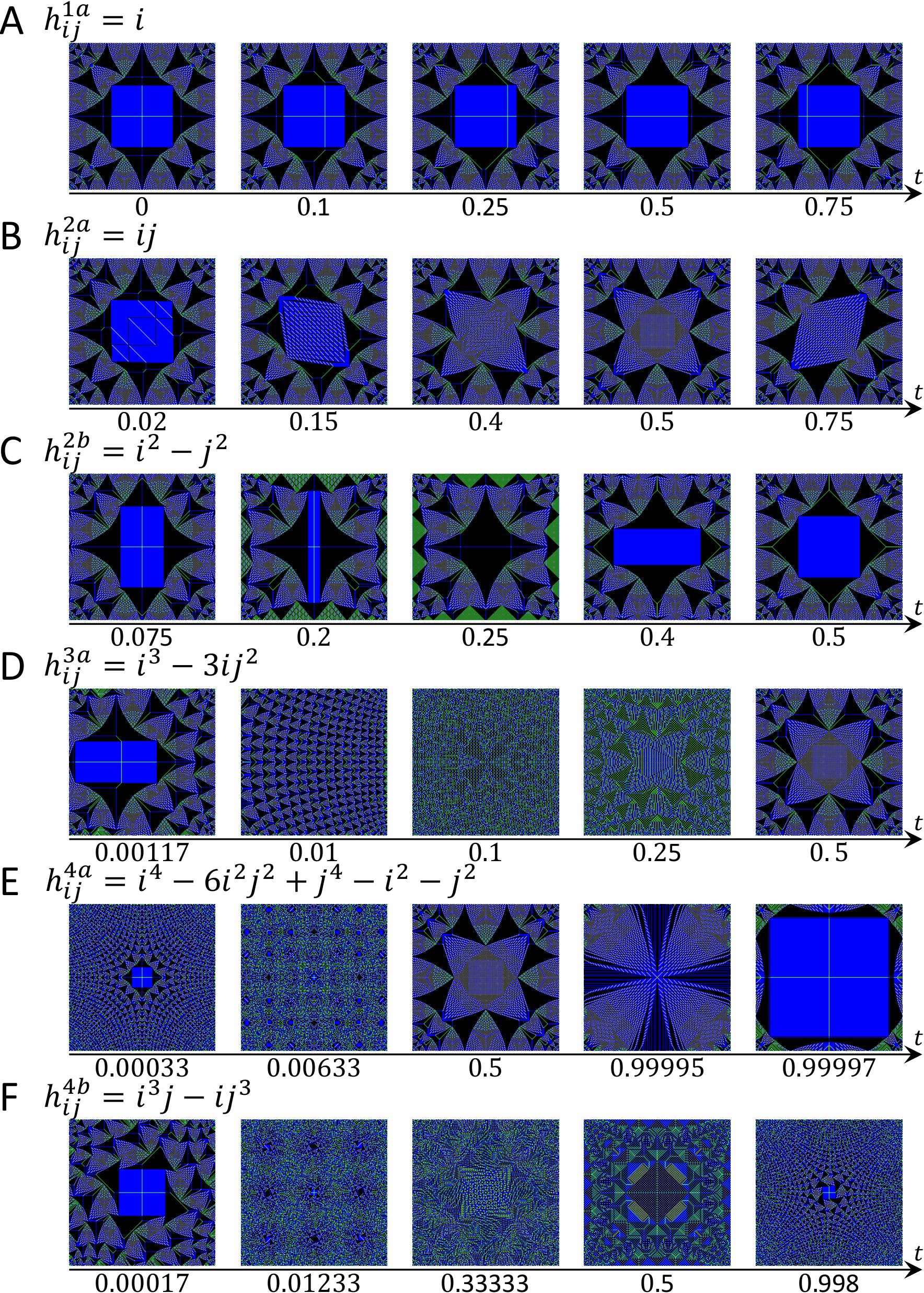}
	\caption{Sandpile identity dynamics induced by harmonic fields of order four or less on a $255\times 255$ square. For odd-order harmonics, only one of the two dynamics is shown since the other one is identical up to switching of axes. For each harmonic, the sandpile dynamics at five time points of specific interest are shown. The first frame ($t=0$) displayed for the dynamics induced by $H^{1a}$ corresponds to the sandpile identity, where the dynamics induced by all harmonics start $(t=0)$ and end $(t=1)$. The six sub-figures correspond to the harmonics A) $h_{ij}^{1a}=i$, B) $h_{ij}^{2a}=ij$, C) $h_{ij}^{2b}=i^2-j^2$, D) $h_{ij}^{3a}=i^3-3ij^2$, E) $h_{ij}^{4a}=i^4-6i^2j^2+j^4-i^2-j^2$, and F) $h_{ij}^{4b}=i^3j-ij^3$.
	}
	\label{Fig2}
\end{figure}

Before discussing the sandpile dynamics, let us first state the following lemma:
\begin{lemma}\label{lemma:periodicity}
The dynamics  $I^H(t)$ of the sandpile identity have periodicity 1, i.e. $I^H(t+1)=I^H(t)$ for all $t\geq 0$.
\end{lemma}
\begin{proof}
The proof follows widely \cite{Creutz1990}, which indeed directly provides the proof for the special case of the constant harmonic function $H^0$. 
As observed in \cite{Creutz1990}, adding four particles to vertex $(i,j)$ and relaxing results in the same configuration as adding one particle to each of its direct neighbors and relaxing. This can be expressed as $a_{ij}^{4}=a_{i-1,j}a_{i+1,j}a_{i,j-1}a_{i,j+1}$, with $a_{ij}$ the operator corresponding to adding one particle to vertex $(i,j)$ and relaxing. Note that we adopt the convention to set $a_{ij}=E$ to the identity operator $E$ for every $(i,j)$ outside of the domain.
 Note that since $I$ is recurrent, $I(t)$ is recurrent for all $t\geq 0$, too. We can thus restrict our analysis to recurrent configurations where $a_{i,j}^{-1}$ exists (\cite{Creutz1990}, Theorem 3), and obtain $a_{ij}^{4}a_{i-1,j}^{-1}a_{i+1,j}^{-1}a_{i,j-1}^{-1}a_{i,j+1}^{-1}=E$ (\cite{Creutz1990}, Eq.~15). 
Similarly, if we add $4\hat{h}_{ij}$ particles to each vertex $(i,j)$, we obtain $\prod_{i,j}a_{ij}^{4\hat{h}_{ij}}a_{i-1,j}^{-\hat{h}_{ij}}a_{i+1,j}^{-\hat{h}_{ij}}a_{i,j-1}^{-\hat{h}_{ij}}a_{i,j+1}^{-\hat{h}_{ij}}=\prod_{i,j}a_{ij}^{-(\Delta\hat{H})_{ij}}=\prod_{i,j}a_{ij}^{x_{ij}}=E$. Finally, $I(t+1) = (I(t)+X^H)^\circ = \left(\prod_{i,j}a_{ij}^{x_{ij}}\right)I(t)=I(t)$.
\end{proof}

In the following, we discuss the sandpile identity dynamics induced by different harmonic functions on a $255\times 255$ square domain. Since on a square, each pair of odd-order harmonic functions results in the equivalent dynamics up to rotation of the sandpile by $90^\circ$, we furthermore only discuss the dynamics induced by one of them. The sandpile identity dynamics for other domains will be discussed at the end of this section.
Because the constant harmonic field $h^0_{ij}=1$ does not lead to any non-trivial dynamics (which directly follows from Eq~\ref{dynDef}), we start the discussion with harmonics functions of order one. In this discussion, we refer to the various two-dimensional patches of the sandpile identity by the names indicated in Figure~\ref{Fig1}A. While, here, we can only show the sandpile identity dynamics at certain times, movies providing a better impression of the continuous dynamics will be made available on public video platforms, and linked on the project's webpage \url{http://langmo.github.io/interpile/}. At the same webpage, an open-source implementation of the algorithms used to generate these movies is available for download.

The sandpile identity dynamics induced by the first order harmonic function $h^{1a}_{ij}=i$ (Figure~\ref{Fig2}A) correspond to a smooth horizontal translation of the tropical curves from left to right, leading to the replacement of each of the three vertical tropical curves -- located originally in the central square and the two adjacent background triangles (Figure~\ref{Fig1}B) -- by their corresponding neighbor to the left after one full period. While these dynamics are rather simple, we note that while a tropical curve passes through a patch, the location of the latter is slightly shifted and its shape slightly changes.  

The sandpile identity dynamics induced by $h^{2a}_{ij}=i j$ (Figure~\ref{Fig2}B) correspond to a smooth ``stretching'' of the central square in the direction of one diagonal, and to a compression in the direction of the other. During this process, the central square gradually changes its pattern by the subsequent action of tropical curves. Thereafter, the central square splits into the tips of two Sierpinski triangles which continue traveling on the diagonal to the corners of the plate. On their way, these tips combine with the rest of the patterns of the Sierpinski triangle which resolved from the inner most Sierpinski triangles on the other diagonal. The tips of the latter continue moving to the center of the domain and eventually form a new central square. Thus, in the course of one period, the self-similar patches of the sandpile identity are smoothly transformed onto each other. These smooth transformations of the patches seem to be accomplished by the repeated action of tropical curves passing through them. 

Similar to the sandpile identity dynamics induced by $h^{2a}_{ij}=i j$, the dynamics induced by $h^{2b}_{ij}=i^2-j^2$ (Figure~\ref{Fig2}C) resemble stretching actions, however, along the horizontal/vertical axes instead of the diagonal ones. Different to the dynamics induced by $H^{2a}$, the central square is ``disassembled'' in horizontal direction from the outside into many tropical curves, and simultaneously grows vertically. The central square thus becomes a rectangular of shrinking width and growing height. At $t=0.25$, the width of the rectangular approaches zero, leading to the fusion of the two background triangles which were originally to the left and the right of the central square. Subsequently, a rectangle is re-established in the center of the domain by the ``fusion'' of many tropical curves entering the domain from the top and the bottom. This rectangular eventually converges to a patch resembling the central square of the identity at $t=0.5$. While the configuration at $t=0.5$ is very similar to the sandpile identity, important features like the two tropical curves crossing in the center of the domain are missing. After $t=0.5$, the dynamics go through another, similar cycle as before, finally reaching the sandpile identity configuration at $t=1$. During each of these two ``half-cycles'', the Sierpinski triangles which were originally above and below the central square leave the domain at its top and bottom boundaries, and are replaced by Sierpinski triangles which were formerly at the diagonals. Similarly, the Sierpinski triangles which were originally at the diagonals are replaced by the ones originally at the left and right of the central square, which are again replaced by new Sierpinski triangles entering the domain at its right and left boundaries. 

For times close to $t=0$ and $t=1$, the dynamics induced by $h^{3a}_{ij}=i^3-3ij^2$ (Figure~\ref{Fig2}D) resemble an horizontal translation of the sandpile identity, overlaid by stretching dynamics similar to the ones induced by $H^{2b}$. Specifically, new Sierpinski triangles enter the domain at its right boundary, leading to regular repetitive configurations of more and more, smaller and smaller Sierpinski triangles filling the whole domain. When the size of each individual Sierpinski triangle approaches one vertex, the dynamics enter extended periods of seemingly random configurations. However, at regular times corresponding to multiples of $1/12$ or other simple fractions, new fractal configurations emerge and subsequently disappear. At $t=1/3$ and $t=2/3$, these fractal configurations are very similar, but not identical, to the sandpile identity (compare $t=0.5$ in Figure~\ref{Fig2}C), and, at $t=0.5$, the configuration is similar but not identical to the one observed for $H^{2a}$ at $t=0.5$ (compare Figure~\ref{Fig2}B). Interestingly, the relative locations at the boundaries where new patches enter and leave the domain in the dynamics induced by $H^{3a}$ seem to be consistent with the ones induced by $H^{2b}$ (see above).

The dynamics induced by $h^{4a}_{ij}=i^4-6 i^2 j^2 + j^4- i^2 - j^2$ (Figure~\ref{Fig2}E) initially seem to ``zoom out'' from the sandpile identity. During this process, more and more patches appear at relative positions seemingly consistent with the ones observed for the dynamics induced by $H^{2b}$ and $H^{3a}$. After some time, several regularly spaced fractal structures resembling the sandpile identity appear, each restricted to a region of the domain. Each of these fractals shrinks in size while converging to the center of the domain. Thus, the domain is slowly filled with dozens or hundreds of such local fractal structures. When the size of each of these fractal structures approaches one vertex, the dynamics enter extended regions of seemingly noisy configurations, except for a small region around the center of the domain which remains relatively regular. From these noisy configurations, accentuated regular fractal structures filling the whole domain emerge at times corresponding to multiples of $1/12$ or other simple fractions. Shortly before these fractal structures emerge, small patches originally located at the center of the domain quickly expand and fill nearly the whole domain with a single, regular pattern. Subsequently, this pattern changes due to the repeated action of tropical curves and then shrinks in size, giving rise to a fractal structure of which the pattern eventually become the center. For example, when approaching $t=1$, four Sierpinski triangles emerge, aligned to the diagonals and with tips touching one another in the domain center ($t=0.99995$ in Figure~\ref{Fig2}E). The dynamics then ``zoom into'' the tips of these triangles, finally resulting in a configuration where the central square of the sandpile identity nearly fills the whole domain ($t=0.99997$ in Figure~\ref{Fig2}E). Subsequently, this central square shrinks and more and more of the other patches constituting the sandpile identity become visible, finally forming the sandpile identity at $t=1$. Due to the observation that in the dynamics induced by $H^{4a}$, local fractal structures appear which converge to the center of the domain and presumably form the regular configurations emerging at times corresponding to multiples of $1/12$ or other simple fraction, we conveniently refer to the latter configurations as hyper-fractals in the following.

The sandpile identity dynamics induced by $h^{4b}=i^3j-ij^3$ (Figure~\ref{Fig2}F) initially also seems to ``zoom out'' from the sandpile identity. These dynamics are however overlaid by some a rotational action around the center of the domain. Initially, this results in a slow rotation of the central square, while all other patches become skewed. Similar to $H^{4a}$, local fractals resembling skewed copies of the sandpile identity enter the domain at its boundaries and slowly converge to its center while shrinking in size. The dynamics then enter extended regions of seemingly random configurations -- again except for a small region around the center of the domain -- interrupted by regular hyper-fractal configurations emerging at times corresponding to multiples of simple fractions. The dynamics finally seem to ``zoom into'' the sandpile identity, instead of ``zooming out'' into the sandpile identity as observed for $H^{4a}$. Interestingly, the hyper-fractal configuration at $t=0.5$ resembles the sandpile identity, however, with most of its self-similar patches being replaced by isosceles triangles (reminiscent of sandpile identities on pseudo-Manhattan domains, see \cite{Caracciolo2008}, Figure~6). This configuration also possesses a central square, however, rotated by $45^\circ$, indicating that a full period leads to the rotation of the central square by $90^\circ$. Finally, we note that the configurations at $t=1/3$ and $t=2/3$ also resemble -- when viewed from afar -- the structure of the sandpile identity. However, these configuration are not composed of patches with clearly defined boundaries, but rather of different patterns which smoothly shade into one another.

Since the computational requirements to determine the sandpile identity dynamics quickly increase with the order of the harmonic field inducing them, we decided to only analyze the dynamics induced by harmonics up to an order of four. Nevertheless, we took a short glimpse into the beginnings of the sandpile identity dynamics induced by the harmonic function $h^{5a}_{ij}=3i^5-30i^3j^2+15ij^4-10i^3$. The dynamics resemble the action of tensile forces, overlaid by zooming actions similar to the ones induced by $H^{4a}$. These tensile forces seem to lead to a normal tensile crack of the central square, i.e. a fracture perpendicular to the left side of the square. Similarly as for fourth order harmonics, also regularly spaced local fractal structures appear in the dynamics induced by $H^{5a}$, eventually leading to seemingly random configurations interrupted by regular configurations at times corresponding to multiples of simple fractions. Despite these initial observations, a detailed description of the dynamics induced by $H^{5a}$ will have to await future research.

Due to the floor function in Eq.~\ref{dynDef}, the sandpile identity dynamics are not linear in the harmonic fields inducing them. Nevertheless, when we analyzed the sandpile identity dynamics induced by linear combinations of different harmonic fields, we observed that the combined actions of two harmonic fields seems to be well described by the sum of their individual actions (Supplementary Figure~\ref{FigS1}A and B). Furthermore, when we applied the harmonic $H^{2a}$ to a configuration extracted from the dynamics induced by $H^{4a}$ at the time when multiple regularly spaced local fractal structures were clearly visible, each of these local fractal structures was separately transformed by stretching actions resembling the dynamics of the sandpile identity under the harmonic field $H^{2a}$ (Supplementary Figure \ref{FigS1}C, compare Figure~\ref{Fig2}B). Similarly, the hyper-fractal configurations occurring in the dynamics of $H^{4a}$ showed dynamics consistent with those of the sandpile identity when induced by $H^{2a}$ or $H^{2b}$ (Supplementary Figure \ref{FigS1}D and E). This indicates that the effect of the non-linearity induced by the floor function in Eq.~\ref{dynDef} is rather negligible, and that the action induced by the sum of two harmonics is well described by the sum of the actions induced by each individual harmonic field.

The sandpile identity dynamics on non-square domains closely resemble those on square domains (Supplementary Figures~\ref{FigS2} and \ref{FigS3}). Interestingly, the dynamics induced by $H^{4a}$ on a rectangular or circular domain (Supplementary Figures~\ref{FigS2}E and \ref{FigS3}E) seem to quickly re-establish the central square and the surrounding patches typical for the sandpile identity on a square domain (Figure~\ref{Fig2}E), even though the identities for these domains share only little similarities. When we further analyzed this effect, we observed that the dynamics induced by $H^{4a}$ seem to be always rather similar, even when we considered non-convex domains or domains having holes (Supplementary Figure~\ref{FigS4}). Specifically, for all tested domains, we saw the emergence of regularly spaced local fractal structures at approximately the same time, and these local fractal structures always resembled the structure of the sandpile identity on a square domain. 

\subsection{Effect of scaling the domain size}
\begin{figure}[htb]
	\centering
	\includegraphics[width=0.73\linewidth]{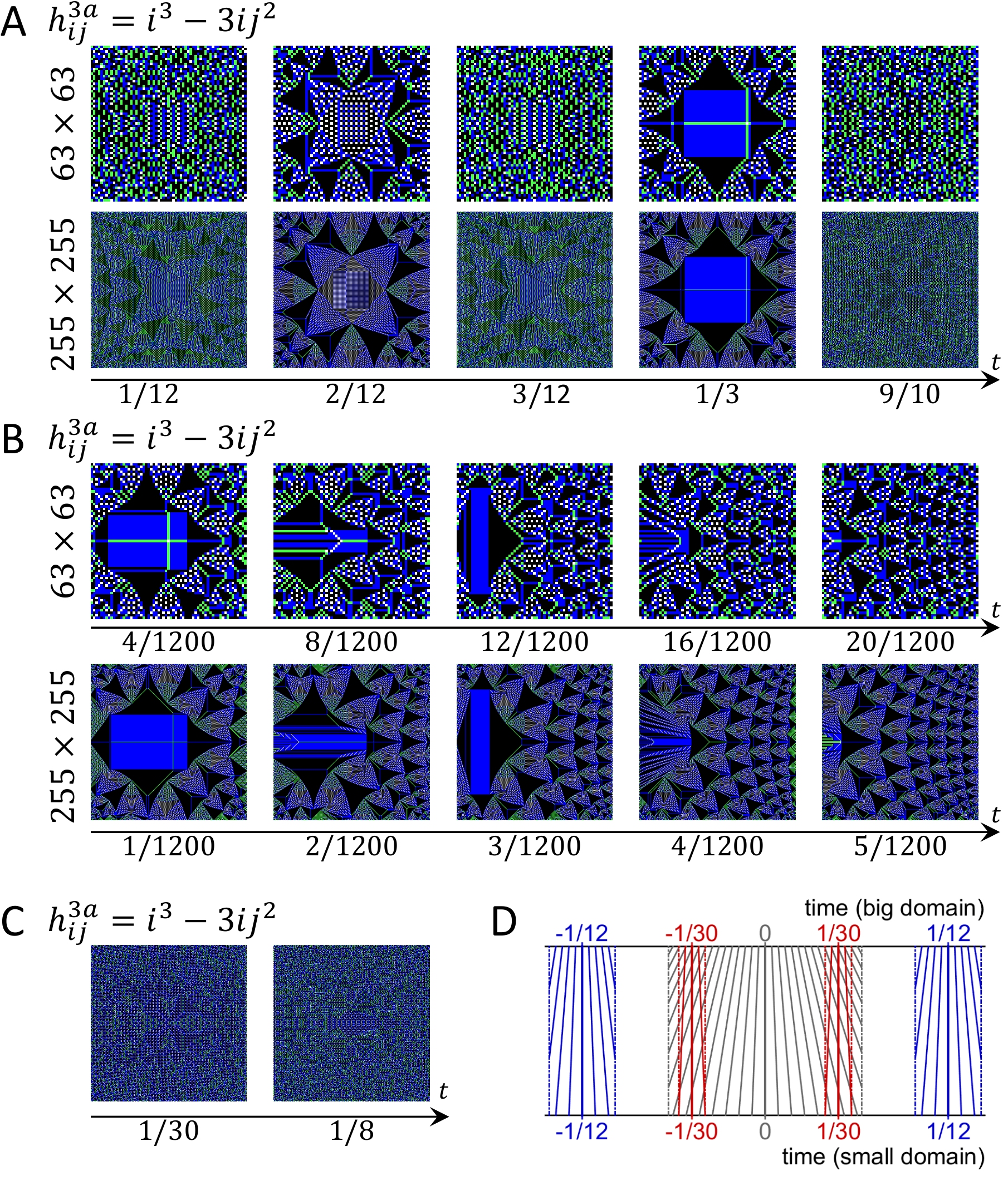}
	\caption{Effect of scaling the domain size on the sandpile identity dynamics induced by $h^{3a}_{ij}=i^3-3ij^2$ (top row: $63\times 63$, bottom row: $255\times 255$). 
	A) For different domain sizes, pronounced fractal configurations occur at the same times corresponding to multiples of simple fractions (here: multiples e.g. of $1/12$). These fractals are easier to visually detect the larger the domain size, and some are only visually detectable for sufficiently large domains (the one at $t=0.9$).
	B) The fractals described in (A) seem to ``entrain'' the dynamics of patches in their temporal vicinity. To map the positions of such entrained patches onto one another for different domain sizes, time has to be scaled by a factor proportional to the domain size (here: by a relative factor of $255/63\approx 4$), which is in contrast to the absolute time scale at which the fractal configurations themselves appear.
	C) For sufficiently large domains, weak fractals can appear in the vicinity of stronger ones (here: at $t=1/30$ and $t=1/8$).
	D) Qualitative model for the interplay between the different effects described in (A-C): A global time scale determines when fractal configurations occur. These fractals impose a time distortion proportional to the domain size in their temporal vicinity. The absolute durations during which these time distortions are ``detectable'' is similar for different domain sizes. The entrainment by weak fractals in the vicinity of stronger ones might lead to a super-position of different local time scales. 
	}
	\label{Fig3}
\end{figure}
In the last section, we discussed the sandpile identity dynamics on a given domain induced by harmonic fields of different orders. In this section, we analyze how the dynamics induced by the same harmonic field change when scaling the domain size. Due to space limitations, we focus our analysis on the dynamics induced by $H^{3a}$ on square domains. We will however generalize our observations to harmonics of different orders at the end of this section.

As shown in the last section, accentuated fractal configurations emerge in the dynamics induced by $H^{3a}$ from seemingly random configurations at times corresponding to multiples of simple fractions (e.g. at $t=1/12, 2/12,\ldots$). In the temporal vicinity of these fractal configurations, the dynamics seem to be ``entrained'' by these fractals, in the sense that new patches emerge at the boundaries of the domain seemingly extending the fractal structure, resulting in regular configurations until the size of these patches approach one vertex. 
To further analyze this effect, we compared the sandpile identity dynamics induced by $H^{3a}$ on $N\times N$ square domains for different domain sizes $N$. In the dynamics on such domains, we observed that similar accentuated fractal configurations appear at the same absolute times corresponding to multiples of $1/12$ or other simple fractions (Figure~\ref{Fig3}A). However, the smaller the domain size and the ``less simple'' the fraction, the less visually pronounced these fractals became ($t=1/12$ or $t=3/12$ in Figure~\ref{Fig3}A). In agreement with this observation, weak fractal configurations at times corresponding to multiples of some other fractions than $1/12$ were only visually detectable for large domain sizes ($t=9/10$, Figure~\ref{Fig3}A). This suggests that fractal configurations always emerge at the same absolute time independently of the domain size, and that even more of such fractal structures might become visible when further increasing the domain size.

A completely different picture emerged when we compared the sandpile identity dynamics for different domain sizes $N$ in the temporal vicinity of strong fractal configurations like the sandpile identity at $t=0$ (Figure~\ref{Fig3}B). Here, the dynamics corresponding to different domain sizes at the same absolute times were clearly different. However, when we scaled the time by a factor $N_1/N_2$ corresponding to ratio of the domain sizes, the dynamics in the vicinity of the strong fractal configurations seemed to coincide (Figure~\ref{Fig3}B). For example, if, for a square domain of size $63\times 63$, time runs faster by a factor of approximately $255/63\approx 4$ than for a domain of size $255\times 255$, the positions and sizes of the individual patches visible in the configurations of the two domains seemed to coincide. 

To interpret this ``time distortion'' in the vicinity of strong fractal configurations, consider that patches entering the domain at its boundary seem to extend the sandpile identity, but that the dynamics can enter periods characterized by seemingly noisy configurations only if the area of each of these individual patches has approached one vertex. Without any ``time distortion'' in the vicinity of strong fractal configurations, we would thus expect, for only large enough domains, to see regular configurations composed of small patches extending over the whole sandpile identity dynamics. In contrast, when we estimated the (absolute) time period during which the entrainment of the sandpile dynamics by the strong fractal configuration at $t=0$ corresponding to the sandpile identity was still observable, this period had similar -- or even the same -- lengths for different domain sizes. The local ``time distortion'' in the proximity of strong fractals thus seems to compensate the effect of the ``space distortion'' caused by scaling the domain size, in the sense that the dynamics enter periods of seemingly noisy configurations at approximately the same absolute times, independently of the domain size.
That there seems to be a ``time distortion'' in the vicinity of strong fractal configurations depending on the domain size becomes even less intuitive when considering that weak fractal configurations can also emerge in the vicinity of stronger ones -- for example at $t=1/30$ which is close to $t=0$ (Figure~\ref{Fig3}C). In Figure~\ref{Fig3}D, we depict a model incorporating all effects observed for the sandpile identity dynamics induced by $H^{3a}$, which proposes a super-position of different time lines to solve this apparent paradox. Due to computational limitations, the rigid evaluation of this model however remains a task for future research.

We may ask if similar effects as observed for the sandpile identity dynamics induced by $H^{3a}$ can also be observe for other harmonic fields. For $H^{1a}$, the sandpile identity dynamics describe the translation of tropical curves (Figure~\ref{Fig2}A). When we compare these dynamics between domains of different size, the tropical curves seem to have the same positions at the same absolute times independent of the domain size (Supplementary Figure~\ref{FigS5}A). The dynamics induced by $H^{2a}$ and $H^{2b}$ map self-similar patches onto one another. When comparing these dynamics between domains of different size, all patches seem to have the same position and shapes at the same absolute times, independently of the domain size (Supplementary Figure~\ref{FigS5}B). However, when we compare the positions of the tropical curves for different domain sizes, their positions were clearly different at the same (absolute) times. Only when we scaled time by a factor proportional to the relative domain size, the positions of the tropical curves coincided, resembling the relationship between the time scale of patch movement/transformation and the appearance of fractal configurations for $H^{3a}$. Finally, for the dynamics induced by fourth order harmonics, the speed of patch movement seems to scale with the square of the domain size. In contrast, for different domain sizes, the speed of the local regularly spaced fractals structures resembling the sandpile identity ($t=0.00633$ in Figure~\ref{Fig2}E) seems to scale with the domain size. Finally, hyper-fractal configurations seem to occur at the same absolute times.

Given these observations, we propose a model in which we assign different dimensions to the different types of ``objects'' appearing in the sandpile identity. Specifically, we assign the dimension $d_c=1$ to tropical curves, $d_p=2$ to patches, $d_f=3$ to fractals, and $d_e=4$ to hyper-fractals. Given an harmonic field of order $o$, an object of dimension $d$ is only ``transformed'' by the corresponding sandpile identity dynamics if $o\geq d$. Given that it is transformed, the speed at which this transformation occurs is proportional to $N^{o-d}$, with $N$ the size of the domain. 

\subsection{The topology and renormalization of the abelian sandpile group}
In this section, we discuss the topology and the renormalization of the abelian sandpile group. Intuitively, for a finite simply-connected domain $\Gamma$, a renormalization assigns some universal coordinates to each configuration of the corresponding sandpile group. With these universal coordinates, it becomes possible to identify configurations of sandpile groups corresponding to different domains, a prerequisite to determine their limiting structure for an infinitely large domain. For our analysis, we denote by $\partial\Gamma$ the boundary of the domain $\Gamma$, and by $G$ the sandpile group corresponding to $\Gamma$. Note that $G$ is generated by adding particles only to vertices at the domain boundary \cite{Creutz1990}, and that the existence and uniqueness for Dirichlet problem on $\Gamma$ holds. The latter implies that, for each combination of particle drops at the domain boundary resulting in the projection of a recurrent configuration on itself after toppling, there exists some corresponding harmonic function $H$ defined both in the interior and exterior of the domain $\Gamma$.

Nonetheless $G$ is a discrete group, the continuity of the harmonic dynamics suggests that it approximates some Lie group $\tilde G$. To analyze the topology of $\tilde G$, we first introduce an extended (modified) sandpile model where each vertex at the boundary of the domain is allowed to carry a non-negative real number of particles, whereas each vertex in the interior of the domain can still only carry a non-negative integer number of particles. The toppling rules for the extended sandpile model are the same as for the original one, i.e. a vertex carrying four or more particles becomes unstable and topples  -- decreasing its number of particles by four and increasing the number of particles of each of its neighbors by one -- independent if it is a vertex at the boundary or the interior of the domain. Similarly, we allow to drop a real number of particles on vertices at the boundary, and a non-negative integer number of particles on vertices in the interior of the domain. It then directly follows that also the space of recurrent configurations of the extended sandpile model forms a topological abelian group $\tilde G$ -- referred to as the extended sandpile group -- of which the original sandpile group $G$ is a discrete subgroup. The quotient group $\tilde G\slash G$ is isomorphic to the torus $(\mathbb{R}\slash\mathbb{Z})^{\partial\Gamma},$ where $\mathbb{R}\slash\mathbb{Z}$ denotes the circle of length $1,$ the result of gluing the ends of the interval $[0,1]$. In particular, $\tilde G$ must be made of cubes $[0,1]^{\partial \Gamma}$, one for each element of $G,$ glued along their boundary faces, which corresponds to a union of tori. In fact, it is a single torus, as the following proposition shows.

\begin{proposition}
 The group $\tilde G$ is a torus of dimension $|\partial\Gamma|$ and volume $|G|$.
 \label{prop_sgtorus}
\end{proposition}

\begin{proof}
Since $(\mathbb{R}\slash\mathbb{Z})^{\partial\Gamma}$ is a torus of volume $1$ and $\tilde G$ is its extension by $G,$ the claims about the dimension and the volume are straightforward provided that $\tilde G$ consists of a single connected component. Recall that $G$ is generated by dropping particles only onto vertices at the boundary $\partial\Gamma$ of the domain $\Gamma$. Therefore, $\tilde G$ is the quotient of $\mathbb{R}^{\partial\Gamma}.$ We complete the proof by recalling that the continuous image of a connected space is connected. 
\end{proof}

The precise relationship between $\tilde G$ and its subgroup $G$ is the following. There is a floor function from $\tilde G$ to $G$. The preimage of a recurrent state $\phi\in G$ under the floor function is a cube of volume $1$ with vertices at $(\phi+\epsilon)^\circ\in G,$ where $\epsilon$ is a function taking $0$ or $1$ and supported on $\partial\Gamma$. 
This relationship is best illustrated on the examples of either a domain consisting of only a single vertex, or a domain consisting of two adjacent vertices. 
In the case when $\Gamma$ consists of just one vertex, all stable configurations are recurrent. Therefore, the sandpile group is $\mathbb{Z}\slash 4\mathbb{Z}$ and the extended sandpile group is $\mathbb{R}\slash 4\mathbb{Z}$, i.e. the circle of length $4$ (see Figure \ref{Fig4}A). 
If $\Gamma$ consists of two adjacent vertices, $15$ out of the $16$ stable configurations on $\Gamma$ are recurrent, with the only non-recurrent configuration being the configuration $c^0_{ij}=0$ where both vertices carry zero particles. The sandpile group $G$ is cyclic with the sandpile identity given by the configuration $c^3_{ij}=3$ where each vertex carries three particles. Since $|\partial\Gamma|=|\Gamma|=2,$ the extended sandpile group $\tilde G$ is a two-dimensional torus of area $15$. It is isomorphic to $\mathbb{R}^\Gamma\slash\Delta\mathbb{Z}^\Gamma$ (Figure \ref{Fig4}B), where $\mathbb{R}^\Gamma\slash\Delta\mathbb{Z}^\Gamma$ denotes the quotient of $\mathbb{R}^\Gamma$ by the image of the Laplacian $\Delta\mathbb{Z}^\Gamma\subset \mathbb{Z}^\Gamma$.

In order to define a proper renormalization, we would like to have universal coordinates on sandpile groups for all $\Gamma$. In the following, we show that those coordinates are provided by the harmonic dynamics described in this paper. Let $\mathcal{H}$ be the space of real-valued discrete harmonic functions on $\mathbb{Z}^2.$ Denote by $\mathcal{H}_\mathbb{Z}\subset \mathcal{H}$ the subspace of integer-valued harmonics. The action $\tilde{C}(t)=(\tilde{C}+tX^H)^\circ$ of a line spanned by $H\in\mathcal{H}$ on a recurrent configuration $\tilde{C}$ of the extended sandpile group $\tilde G$ defines a geodesic on $\tilde G$. If $H\in \mathcal{H}_\mathbb{Z}$, the geodesic is closed. The projection $C(t)=\lfloor\tilde{C}\rfloor$ of this geodesic to $G$ via the floor-function $\lfloor.\rfloor$ gives the cyclic trajectory of the corresponding harmonic sandpile dynamics (compare Eq.~\ref{dynDef}). 

Recall how we constructed the sandpile identity dynamics for a given harmonic field $H\in\mathcal{H}_\mathbb{Z}$ (Figure~\ref{Fig1}C), which can be extended to an homomorphism $\eta:\mathcal{H}\rightarrow\tilde G$ mapping harmonic fields onto configurations of the extended sandpile group as follows: For a given harmonic $H\in\mathcal{H}$ defined on the standard square lattice $\mathbb{Z}^2$, we first define its restriction $\tilde H=H|_\Gamma$ to the domain $\Gamma$. 
We then determine the minimal number $k\in\mathbb{Z}$  such that $\Delta (\tilde H+k)$ becomes non-positive. With these definitions, the homomorphism (corresponding to the extended sandpile identity dynamics) is given by 
\begin{align*}
\eta(H)=(-\Delta(\tilde H+k)+I)^\circ.
\end{align*}
Note that since $H$ is not restricted to take integer values anymore, an additional time parameter as in Eq.~\ref{dynDef} is not required. Similarly, the fact that vertices can take real values at the boundary in the extended sandpile model allows us to omit the floor function.

The following lemma is a straightforward consequence of the existence of a solution to the Dirichlet problem on $\Gamma$ and outside \cite{Lawler2010}: 
\begin{lemma}
The homomorphism $\eta$ is surjective.
\label{lem_etasurj}
\end{lemma}

Now consider an exhausting injective family of domains $\{\Gamma_N\subset\mathbb{Z}^2\},$ i.e. $\Gamma_N\subset\Gamma_{N+1}$, where the union of all $\Gamma_N$ is $\bigcup_N\Gamma_N=\mathbb{Z}^2.$ Let $\tilde G_N$ be the extended sandpile group of $\Gamma_N$.
\begin{theorem}
There are canonical surjective homomorphisms $\tilde G_{N+1}\rightarrow \tilde G_{N}$ with $$\tilde G_\infty:=\projlim_{N\rightarrow\infty}\tilde G_N$$ independent from the specific choice of the family $\{\Gamma_N\}.$ There is a natural inclusion $$\mathcal{H}\slash\mathcal{H}_\mathbb{Z}\rightarrow\tilde G_\infty.$$
\end{theorem}

\begin{proof}
Let $K_N$ be the kernel of $\eta_N:\mathcal{H}\rightarrow\tilde G_N$ This kernel $K_N$ is equal to the space of harmonic functions taking integer values on $\Gamma_N.$ Therefore, $K_{N+1}\subset K_{N}$ which give projections between quotients $$\mathcal{H}\slash K_{N+1}\rightarrow \mathcal{H}\slash K_{N}.$$ By Lemma \ref{lem_etasurj} these quotients  are canonically isomorphic to $\tilde G_{N+1}$ and $\tilde G_N$. Finally, we note that the intersection  of all $K_N$ is $\mathcal{H}_\mathbb{Z}.$ 
\end{proof}

This theorem directly poses the following question, which we hope to answer in a future article:
\begin{question}
Is the inclusion $\mathcal{H}\slash\mathcal{H}_\mathbb{Z}\rightarrow\tilde G_\infty$ an isomorphism?
\label{quest_spglim}
\end{question}

If the answer to this question is affirmative, we could say that the limit of the extended sandpile group on the whole lattice is described by the quotient of real-valued discrete harmonic functions by the subspace of integer-valued ones. Otherwise, to compute the limit of the sandpile group, we would need to describe the quotient $\tilde G_\infty \slash (\mathcal{H}\slash\mathcal{H}_\mathbb{Z}).$ In any case, we hope that thinking of a sandpile group $G$ as a discrete object approximating $\tilde G_\infty$ would induce a better understanding of both, leading to the rigorous description of the sandpile identity and its harmonic dynamics.

Since $\mathcal{H}\slash\mathcal{H}_\mathbb{Z}$ projects onto all sandpile groups $\tilde G$ corresponding to finite domains $\Gamma$, the theorem thus provides a set of universal coordinates for the extended sandpile groups corresponding to all domains $\Gamma$. Since there exists a natural inclusion of $G$ into $\tilde{G}$, these coordinates restricted to $G$ are also universal coordinates for the original sandpile groups corresponding to all domains $\Gamma$. The existence of these universal coordinates directly implies that both the original  and the extended sandpile groups admit a natural renormalization. In more detail, for a pair of domains $\Gamma_1\subset\Gamma_2$, there is a canonical renormalization projection from $\tilde G_2$ to $\tilde G_1.$ Combining this projection with the inclusion $G_2\rightarrow \tilde G_2$ and the floor function $\tilde G_1\rightarrow G_1$ we get the natural renormalization map from $G_2$ to $G_1.$  It amounts to define the projective limit $G_\infty$ which, in contrast to $\tilde G_\infty,$ seems to depend on a particular family of exhausting domains $\{\Gamma_N\}.$ One more circumstance that makes $G_\infty$ even more subtle is that the renormalization map on the original sandpile groups is not a homomorphism and thus the limit should not be a group. Investigation of the renormalization properties of the sandpile group is subject of future research.

\subsection{Tropical Geometry and Super-Harmonic Functions}\label{sec:topology}
While tropical geometry is still a relatively young mathematical field, it was already applied in the study of the BTW sandpile model leading to the identification and mathematical description of the thin ``curves'' or ``strings'' visible in the sandpile identity (Figure~\ref{Fig1}B) as being tropical curves \cite{Caracciolo2010,Kalinin2015,Kalinin2016,Kalinin2017}. In tropical algebra (see e.g. \cite{Brugalle2015,Kalinin2018} for an introduction), the addition operator $\oplus$ corresponds to the classical minimum, while the multiplication operator $\otimes$ corresponds to the classical sum:
\begin{align*}
x\oplus y&=\min(x,y)\\
x\otimes y&=x+y.
\end{align*}

Given that the sandpile identity dynamics induced by harmonic fields seem to transform patches, fractals and hyper-fractals by the repeated action of tropical curves, we asked if tropical geometry can also be directly applied to our study of harmonic fields.
When applied to harmonic fields, tropical multiplication states that the sum of two harmonic functions must be harmonic, too, which is trivially satisfied by the linearity of the discrete Laplacian operator. The tropical sum of two harmonics, however, is in general not a harmonic, but a super-harmonic function, i.e. a function $H^S$ for which $(\Delta H^S)_{ij}\leq 0$. We thus decided to also analyze the evolution of the sandpile identity under super-harmonic fields. Note that all algorithms introduced above for the study of harmonic fields can be directly applied to super-harmonic fields, with the only difference that the potential $X=-\Delta H^S$ can become non-zero also in the interior of the domain.

In Figure~\ref{Fig4}C, we depict the sandpile identity dynamics induced by a super-harmonic field obtained by taking the tropical sum of the three harmonic functions $ij+c_1$, $-3ij+c_2$, $ki+c_3$ and $-3ki+c_4$, with $k=128$ for a $255\times 255$ square domain. The constants $c_1,\ldots c_4$ were chosen such that the square domain is split into three rectangular regions, each governed by one of the harmonic functions. For illustrative purposes, we furthermore chose $c_1,\ldots c_4$ such that each of these rectangular regions has the same width (half the width of the domain), but that the heights of vertically adjacent regions have a ratio $3:1$ (three quarter, respectively one quarter, of the domain height). 
The sandpile identity dynamics of each region widely resemble the dynamics induced by its respective harmonic field alone. Specifically, the patches in the right half of the domain, which is governed by first order harmonics only transforming tropical curves, remain nearly fixed in space. In contrast, the dynamics of the top-left and bottom-left regions of the domain resemble the stretching actions of $H^{2a}$. Due to the opposite signs and different magnitudes of their governing harmonics, the stretching actions in these two regions however have opposite directions and different speeds.

The partitioning of the domain into distinct regions each showing sandpile identity dynamics corresponding to the ones of its governing harmonic field is, however, not perfect: We can clearly observe tropical curves which extend over more than one of these regions. These tropical curves result in a cross-talk between the regions disturbing the structure of their respective patches. That such tropical curves spanning multiple regions occur is not surprising: there exist several well known limitations for allowed local patterns in recurrent configurations. For example, it is not possible that two neighboring vertices of a recurrent configuration both carry zero particles \cite[p.17]{Creutz1990,Paoletti2014}. Intuitively, a strict division of the sandpile dynamics into regions of the domain governed by different harmonic functions would, at one time or the other, violate one of these conditions; thus, a ``perfect partitioning'' of the domain into such regions can in general not exist -- at least not for finite domains. 
Nevertheless, already the observation that super-harmonic functions lead to the partitioning -- even if not perfect -- of the domain into distinct regions each displaying dynamics typical for the harmonic field governing that region is rather surprising, and suggests that tropical geometry might apply to more than the study of tropical curves in the BTW sandpile model.

\subsection{Stochastic realizations of harmonic potentials}
\begin{figure}[htb]
	\centering
	\includegraphics[width=0.68\linewidth]{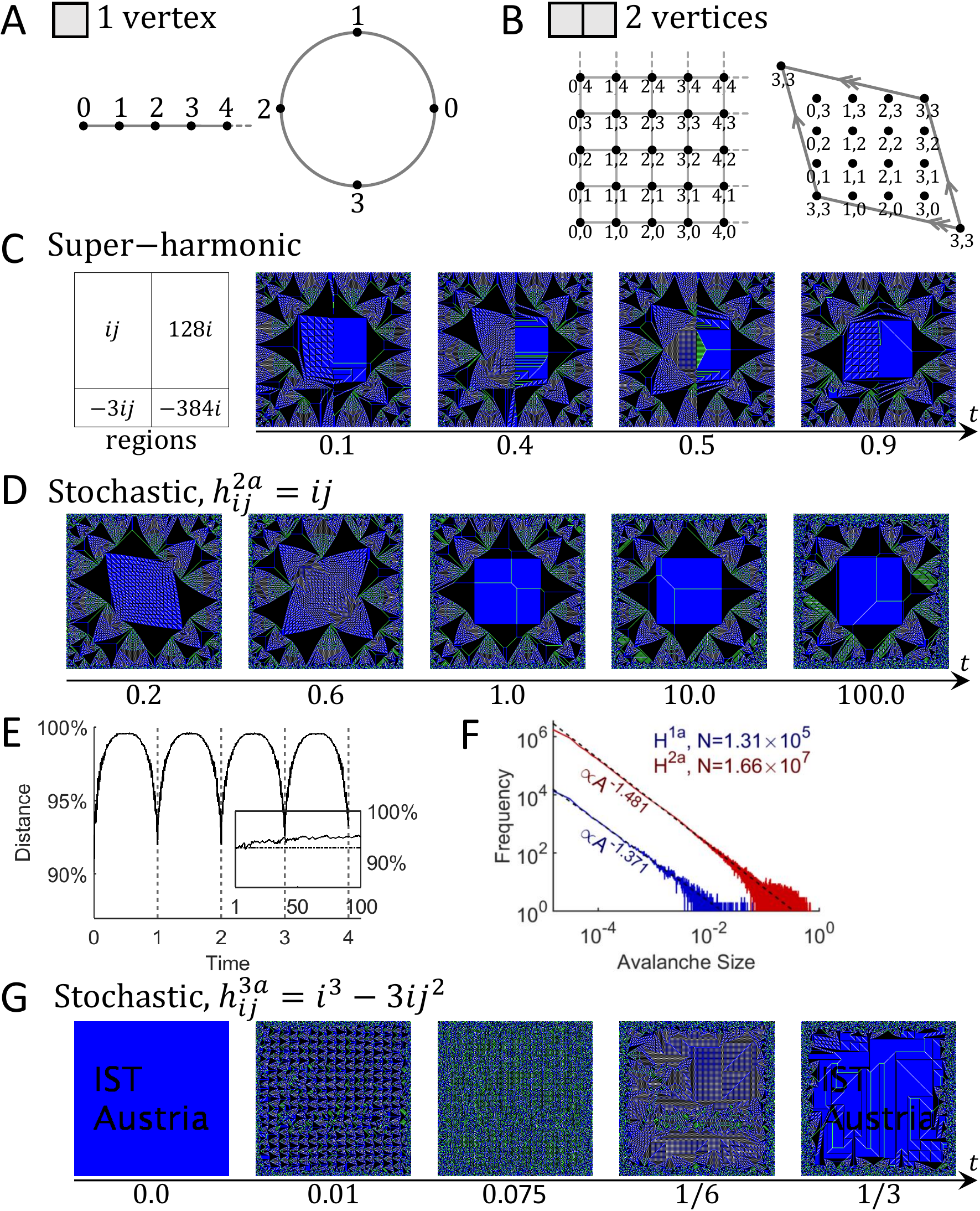}
	\caption{
	A) Topology of the sandpile and the extended sandpile group on a domain consisting of a single vertex. Left: Space of legal (non-negative) states. Right: The corresponding sandpile groups $G=\mathbb{Z}\slash 4\mathbb{Z}\subset \tilde G=\mathbb{R}\slash 4\mathbb{Z}$.
	B) Space of legal states (left) and the sandpile group (right) on a domain consisting of two adjacent vertices. The extended sandpile group is a two-dimensional torus obtained by gluing the opposite sides of the rhombus, the lattice points on $\tilde G$ form $G$. The corner points all  represent the same point which is the sandpile identity. Moving along horizontal/vertical directions corresponds to increasing the amount of sand carried by one of the two vertices.
	C) Sandpile identity dynamics induced by a super-harmonic function, dividing the domain into four connected rectangular regions governed by the harmonic fields $ij$ (top-left), $-3ij$ (bottom-left), $ki$ (top-right), and $-3ki$ (bottom-right), with $k=128$ for a $255\times 255$ domain.
	D) Stochastic realization of the sandpile identity dynamics induced by the field $H^{2a}$. 
	E) Evolution of the distance (normalized variation of information) between the stochastic identity dynamics shown in (D) and the sandpile identity. The main axes show the evolution over four periods, the inset the evolution at $t=1,2,\ldots$ for 100 periods.
	F) Avalanche size distribution over one full period of the stochastic identity dynamics induced by $H^{1a}$ (blue) and $H^{2a}$ (red). The dashed black lines show the power law distributions with critical coefficients $-1.371$ for $H^{1a}$ and $-1.481$ for $H^{2a}$, respectively.
	G) Sandpile dynamics induced by higher order harmonic fields (here: $H^{3a}$) allow to encode and decode information into and from seemingly random configurations ($t=0.075$).
	}
	\label{Fig4}
\end{figure}

Recall that, to create a harmonic field inducing the sandpile identity dynamics, we drop particles at regular time intervals onto boundary vertices according to the values of the potential $X$ at the respective vertices (Eq.~\ref{dynDef}). Instead of dropping these particles deterministically, we can alternatively drop them stochastically according to a probability distribution proportional to the potential $X$. This algorithm corresponds to a Markov process resulting in stochastic sandpile identity dynamics with the same expected time interval between two successive particle drops onto the same boundary vertex as the (fixed) time interval of its deterministic counterpart. Due to the simplicity of this Markov process -- the probability to drop a particle onto a given boundary vertex is independent of the current configuration -- it would be relatively simple to implement this algorithm experimentally in a physical system resembling the abelian sandpile. 

To analyze the sandpile dynamics under such a stochastic Markov chain, we dropped particles one-by-one onto boundary vertices, with the probability that a given particle is dropped onto the boundary vertex $(i,j)$ being proportional to $x_{ij}$. After each particle drop, we relaxed the sandpile and associated the time $t=\frac{k}{|X|}$ to the resulting configuration, with $k=0,1,2,\ldots$ the number of particles already dropped.
In Figure~\ref{Fig4}D, we depict the stochastic sandpile identity dynamics corresponding to a realization of this Markov process for the harmonic field $H^{2a}$. The stochastic dynamics closely resemble their deterministic counterparts, but are overlaid by ``noise'' arising due to the stochastic nature of the process. This noise seems to be stronger towards the boundaries of the domain where particles are dropped, while in the interior of the domain the stochastic sandpile dynamics remain very close to their deterministic counterparts besides small changes in the positions of the tropical curves. 

The stochastic sandpile identity dynamics are surprisingly robust: after one period, nearly all details of the sandpile identity are reproduced ($t=1$ in Figure~\ref{Fig4}D). Only after significantly longer time, the structure of the sandpile identity becomes significantly disturbed ($t=10$ and $t=100$ in Figure~\ref{Fig4}D). Specifically, the noisy region towards the boundaries of the domain seems to slowly increases in size over time.
To mathematically analyze this drift, we determined the normalized variation of information between the sandpile identity and the configurations emerging in the stochastic dynamics at different times. The normalized variation of information $\mathcal{VI}(I;I(t))=1-\frac{\mathcal{I}(I,I(t))}{\mathcal{H}(I,I(t))}$ is a distance measure quantifying the difference between two configurations based on statistics on the relationship between the numbers of particles carried by the same vertex in the two configurations, with $\mathcal{I}$ the mutual information and $\mathcal{H}$ the joint entropy. Thus, the normalized variation of information is sensitive with respect to small shifts in the position of periodic patterns constituting e.g. the Sierpinski triangles, which can result in rather high distances assigned to pairs of visually very similar configurations. However, we are not aware of any other reliable distance measure which is more robust with respect to such local shifts.

Due to the reasons described above, the variation of information quickly increased for the stochastic sandpile identity dynamics induced by $H^{2a}$ (Figure~\ref{Fig4}E), and its dynamic range was restricted to the interval between $0.9$ and $1.0$ for all but the smallest times (Figure~\ref{Fig4}E). Nevertheless, it showed clear $1$-periodic dynamics, with local minimums located at integer valued times as expected from the construction of the stochastic process. When we determined the evolution of the variation of information at only these integer valued times for up to 100 periods, we observed a slow drift of the stochastic realization away from its deterministic counterpart (Figure~\ref{Fig4}E, inset). The existence of this slow drift is not surprising due to the following corollary:
\begin{corollary}
Assume that, for a rectangular $N\times M$ domain, all vertices of the potential $X^H$ carry a positive particle number at one of the four domain boundaries, i.e. $x_{1,j}>0$, $x_{N,j}>0$, $x_{i,1}>0$ or $x_{i,M}>0$ for all $j=1,\ldots,M$, respectively $i=1,\ldots,N$. Then, the stochastic sandpile identity dynamics induced by its corresponding harmonic field $H$ are ergodic.
\end{corollary}
This corollary follows directly from the proof in \cite{Creutz1990} that any recurrent configuration can be reached from any other by only dropping particles onto vertices of one of the boundary sides of the domain. For completeness, we quickly restate this proof in the notation used in this article:
\begin{proof}
Let $a_{ij}$ be the operator corresponding to adding one particle to vertex $(i,j)$ and relaxing. If we restrict ourselves to recurrent configurations, then the inverse $a^{-1}_{ij}$ of $a_{ij}$ exist. Furthermore, since the set of recurrent configurations is finite, there exists a number $n_{ij}>0$ such that $a_{ij}^{n_{ij}}=E$, with $E$ the identity operator \cite{Creutz1990}. This implies that adding $n_{ij}-1$ particles to the vertex $(i,j)$ and relaxing results in the same configuration as applying the inverse of $a_{ij}$, i.e. $a^{-1}_{ij}=a_{ij}^{n_{ij}-1}$.
Recall that adding four particles to vertex $(i-1,j)$ and relaxing results in the same configuration as adding one particle to each of its neighbors (see proof of lemma~\ref{lemma:periodicity}): $a_{i-1,j}^4=a_{i,j}a_{i-1,j+1}a_{i-1,j-1}a_{i-2,j}$.
Solving for $a_{i,j}$, we obtain $a_{i,j}=a_{i-1,j}^4a_{i-1,j+1}^{-1}a_{i-1,j-1}^{-1}a_{i-2,j}^{-1}=a_{i-1,j}a_{i-1,j+1}^{n_{i-1,j+1}-1}a_{i-1,j-1}^{n_{i-1,j-1}-1}a_{i-2,j}^{n_{i-2,j}-1}$. Thus, dropping a particle onto a given vertex and relaxing is equivalent to dropping (a specific amount of) particles only onto vertices strictly to the left of the given vertex. By induction, this implies that one can achieve the same effect as dropping a particle on any given vertex by only dropping particles onto the left boundary of the domain. The same argument holds for the other three boundaries of the domain.
\end{proof}

Intuitively, ergodicity means that, for a finite domain, the stochastic sandpile identity dynamics will visit every recurrent configuration of the sandpile group infinitely often. Thus, sooner or later, such stochastic dynamics have to leave the vicinity of the deterministic sandpile identity dynamics, and travel through qualitatively different sets of configurations. While this explains the drift of the stochastic sandpile identity dynamics (Figure~\ref{Fig4}E, inset), it does not explain why this drift seems to be rather slow. We note that this robustness of the sandpile identity dynamics furthermore seems to increase with increasing domain size.

Different to their deterministic counterparts where, at some times, more than one particle is dropped onto the domain at once, each particle is dropped separately to induce the stochastic sandpile identity dynamics. This allows us to determine the critical exponent of the avalanche sizes -- the distribution of the number of topplings after dropping one particle -- for the stochastic sandpile identity dynamics induced by each harmonic (Figure~\ref{Fig4}F). Our results show that the critical exponent is different for different harmonic fields, indicating that there might exist different subsets of the sandpile group each showing a different critical exponent.  

\subsection{Encoding information in harmonic sandpile dynamics}
When examining the proof for Lemma~\ref{lemma:periodicity} for the $1$-periodicity of the sandpile identity dynamics induced by harmonic fields, one can observe that the lemma can be directly extended to sandpile dynamics induced by harmonic fields when not starting with the sandpile identity, but with any other recurrent configuration. For example, in Supplementary Figure~\ref{FigS7}, we show the periodic dynamics induced by $H^{2a}$ and $H^{2b}$ when starting with the ``maximally filled'' configuration $c^3_{ij}=3$ where each vertex initially carries three particles, or when starting with the configuration $c^2_{ij}=2$ where each vertex initially carries two particles. 
In some of these dynamics (Supplementary Figure~\ref{FigS7}B), we see the emergence of structures similar to the Sierpinski triangles of the sandpile identity, whereas other dynamics tend to show other kinds of patterns (Supplementary Figure~\ref{FigS7}A,C\&D).
In the following, we asked if this effect can be utilized to encode information in seemingly random sandpile configurations, and to reliably decode them again.

We assume that the information which should be encoded is given in the form of a binary image $\hat{P}\in\{0,1\}^{N\times M}$. To encode this information, we utilize that, on every rectangular domain, the configuration $c^2_{ij}=2$ where each vertex carries two particles is recurrent. This also implies that the stable configuration corresponding to the sum $P=C^2+\hat{P}$ of $C^2$ and $\hat{P}$ is recurrent, too. The dynamics induced by any harmonic field when starting with the configuration $P$ thus have to be $1$-periodic.
In Figure~\ref{Fig4}G, we show the stochastic harmonic sandpile dynamics induced by $H^{3a}$ starting from a configuration encoding the string ``IST Austria'' as described above. Similar to the sandpile identity dynamics induced by $H^{3a}$, these dynamics pass through extended regions of seemingly random configurations ($t=0.075$ in Figure~\ref{Fig4}G) where neither the encoded information nor the fact that information is encoded at all is apparent. Given such a seemingly random configuration, one can however reconstruct the original information by simply determining the sandpile dynamics induced by the same harmonic field starting from this seemingly random configuration ($t=1/3$ in Figure~\ref{Fig4}G). Notably, in this decoding process it is not necessary to know at which exact time the seemingly random configuration was extracted. Instead, it is sufficient to induce dynamics using the same harmonic field and to wait until the encoded information eventually becomes clearly visible.

\section{Discussion}
While we proved that the sandpile identity dynamics have to be $1$-periodic (Lemma~\ref{lemma:periodicity}), we cannot yet provide a mathematical derivation for our observation that the sandpile identity dynamics correspond to smooth transformations of tropical curves, patches, fractals and hyper-fractals, or that most patches seem to be conserved between adjacent frames or even throughout the whole dynamics. However, there exists some ``hints'' in the literature: it was shown that, while $1$-dimensional defects of the background are tropical curves, i.e.  governed by piecewise-linear functions, patches correspond to polynomials of order two \cite{Levine2016,Pegden2017}. These results are in agreement with our observation that first order harmonic fields seem to only transform tropical curves, while second order harmonic fields also transform patches. Moreover, in the course of the dynamics the patches are assembled and translated by tropical curves. If it is possible to ``formalize these hints'' into a consistent mathematical theory, which would also explain the different scaling laws for the speed with which tropical curves, patches, fractals and hyper-fractals move during the sandpile dynamics, remains a task for future research.
Another open question concerns the apparent high robustness of the sandpile identity dynamics. This robustness did not only manifest itself when we compared the deterministic sandpile identity dynamics with their stochastic counterparts realized by ergodic Markov processes (Figure~\ref{Fig4}D-G), but also when we tested various modifications of our algorithm to determine the sandpile identity dynamics. For example, it seems to make little difference if we replace the floor function in Eq.~\ref{dynDef} by the ceil or the round function (Supplementary Figure~\ref{FigS8}A), or if we distinguish from which cardinal direction particles are dropped onto the sandpile  (Supplementary Figure~\ref{FigS8}B). While the robustness of the sandpile identity dynamics with respect to such modifications is certainly interesting, we note that the floor function is the obvious choice since it leads to a direct relationship between the sandpile identity dynamics of the original and the extended sandpile model (Section~\ref{sec:topology}). We also note that, different to the original model, the sandpile identity dynamics of the extended sandpile model are linear in the harmonics. This suggests that one might always calculate the dynamics on the extended model, and only apply the floor function at the time of visualization.

We expect that our results showing the existence of smooth sandpile identity dynamics induced by harmonic fields provides important impulses for future studies of the abelian sandpile in specific, and of (self-organized) critical systems in general. The emergence of new patches at positions of the domain boundaries consistent between the dynamics induced by different harmonic fields proposes that it might be possible to extend the fractal structure of the sandpile identity beyond the boundaries of the domain, somewhat similar to the continuation of analytic functions in the complex plane. Provided that the dynamics induced by the harmonic $H^{4a}$ (Figure~\ref{Fig2}E) are indeed similar to zooming actions, this extension might result in a (potentially infinite) wallpaper-like structure composed of (locally) similar fractal structures aligned on a regular grid, and this extension might be self-similar to itself when viewed at different magnification levels. The similarities between the sandpile identity dynamics induced by $H^{4a}$ on different domains (Supplementary Figure~\ref{FigS4}) furthermore proposes that there might be only one fundamental extension for each domain topology (topological invariance), or even only one fundamental extension at all. Given such a fundamental extension, a specific extension for a given domain might then be obtainable by some smooth transformation. 

The different scaling laws for the speed with which tropical curves, patches, fractals and hyper-fractals move during the sandpile dynamics seem to be a direct consequence of the fact that tropical curves on finite domains have a finite instead of an infinitesimal width (usually one vertex) which remains unchanged when scaling the domain, while the width of patches visible in the sandpile identity (Figure~\ref{Fig1}B) scales with the domain size.
In this context, it is interesting to note that the smooth transformation of the central square of the sandpile identity induced by second order harmonics into either the tips of two Sierpinski triangles ($H^{2a}$, Figure~\ref{Fig2}B), or into a large set of tropical curves ($H^{2b}$, Figure~\ref{Fig2}C) proposes that all patches constituting the sandpile identity might be composed of tropical curves. In contrast, the dynamics induced by $H^{4a}$ (Figure~\ref{Fig2}C) propose the alternative interpretation that the sandpile identity itself might be composed of fractals. These two interpretations don't necessarily have to contradict each other if we  e.g. would allow tropical curves to be composed of fine strings of fractals. We note that such interpretations are somewhat reminiscent of similar discussions in string theory, where strings (tropical curves) and branes of different dimensions (corresponding to patches, fractals, hyper-fractals and so on) occur. 

Our results also suggest that the harmonic fields themselves generate the extended sandpile group on any domain. We show that the space of harmonics is closely related to the scaling limit of the sandpile group (see Theorem 1). Specifically, for each order, there are two linearly independent harmonics which generate a subtorus of the extended sandpile group. It seems that the directions on this torus are clearly visible in the identity dynamics induced by the harmonics, and are consistent when we linearly combine two harmonics of the same order (Figure S2). For a finite domain, we would furthermore expect that the potentials corresponding to harmonics of sufficiently high order should at one point be linear combinations of the potentials corresponding to lower order harmonics. The minimal action principle for the toppling function \cite{Fey2010} might then imply that, for finite domains, only sandpile identity dynamics for harmonics up to a certain maximal order can be induced, in agreement to the finite number of generators of the abelian group \cite[p.~26--27]{Paoletti2014}. 
Given this interpretation, the observation that regular fractal configurations occur at times corresponding to multiples of simple fractions in the sandpile identity dynamics induced by third or higher order harmonics would have an interesting explanation: these simple fractions would correspond to simple roots of the identity corresponding to the respective generator/harmonic.

Our analysis of the stochastic sandpile identity dynamics induced by harmonic fields generated by Markov processes indicates that also the configurations visited by those dynamics are critical (Figure~\ref{Fig4}F). However, the critical exponent of the avalanche size distributions seems to depend on the order of the harmonic.
This indicates that the different harmonics might divide the sandpile group into different sub-sets each showing scale-free spatio-temporal relationships, however, with different critical exponents. If it would be possible to determine the critical exponents of these individual sub-sets of the sandpile group with a high confidence utilizing the fact that the sandpile identity dynamics are periodic, it might be possible to reconstruct the critical exponent corresponding to the whole sandpile group by taking an adequately weighted mean.

Finally, it seems fascinating that the sandpile dynamics induced by harmonic fields can be utilized to encode information in seemingly random configurations and to robustly decode them utilizing an extremely simple stochastic Markov process requiring only the knowledge of the harmonic function used for encoding.
While we expect the ``randomness'' of the configurations carrying the information to increase with increasing order of the harmonic function used for encoding, and the ``errors'' due to the stochasticity of the decoding mechanism to decrease with increasing domain sizes, further research would have to analyze if this mechanism results in a sufficiently safe information encoding. Even though a widespread application of the sandpile model in cryptography is rather unlikely, our results pose the question if it is possible to implement similar processes for other self-critical or critical systems. If so, our results might lead to new information storage mechanisms when applied to physical systems, or new interpretations when applied to naturally evolved biological networks.

\bibliographystyle{IEEEtran}
\bibliography{sandpiles}

\clearpage

\section*{Supplementary Figures}
\setcounter{figure}{0}
\renewcommand{\figurename}{Supplementary Figure}
\renewcommand{\thefigure}{S\arabic{figure}}
\begin{figure}[htb]
	\centering
	\includegraphics[width=0.73\linewidth]{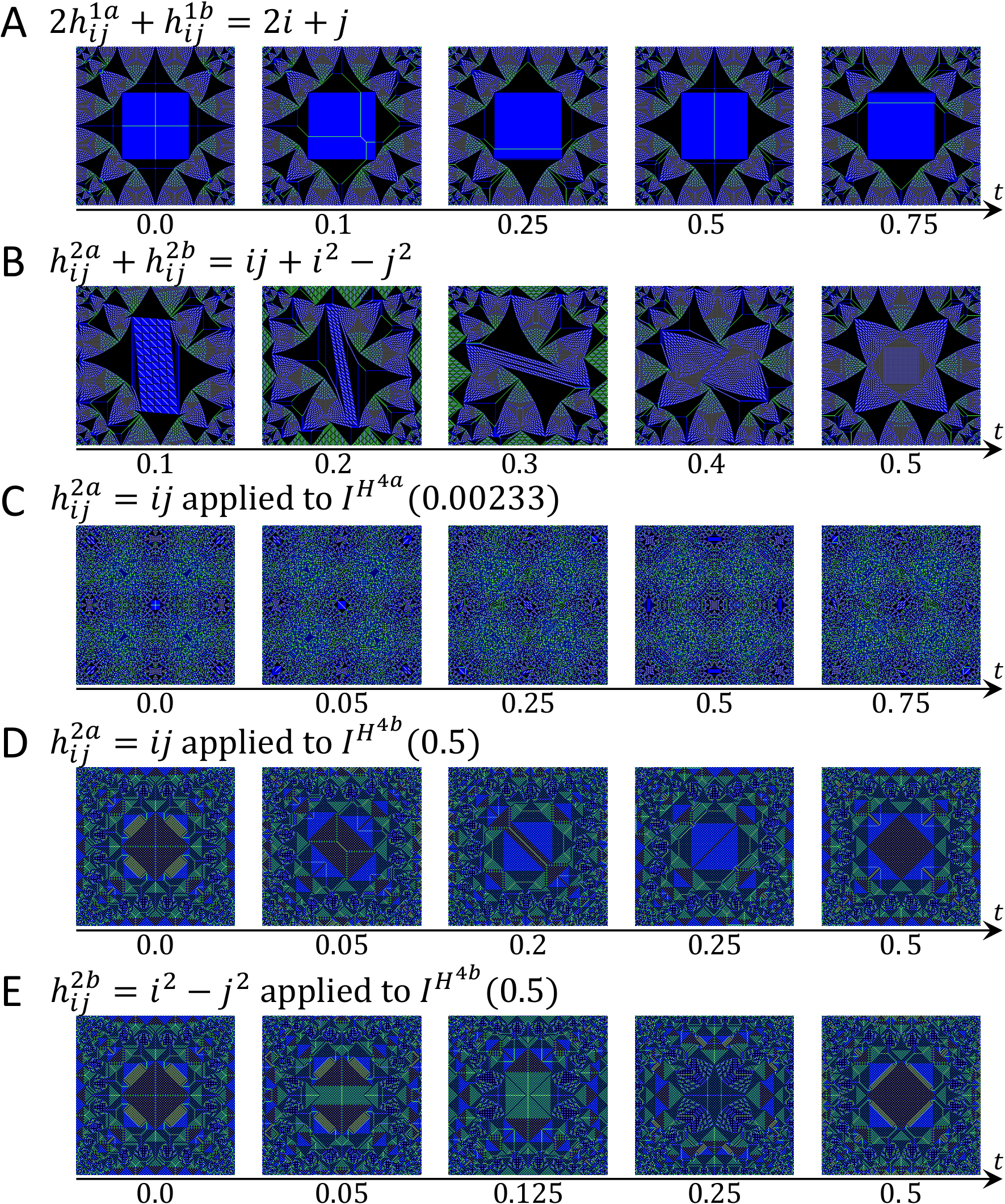}
	\caption{
	Sandpile identity dynamics on a $255\times 255$ domain induced by linear combinations of the harmonic fields used in Figure~\ref{Fig2}. A) Dynamics induced $2h^{1a}_{ij}+h^{1b}_{ij}=2i+j$. B) Dynamics induced by $h^{2a}_{ij}+h^{2b}_{ij}=ij+i^2-j^2$. C) Dynamics induced by $h^{2a}_{ij}=ij$ starting at a configuration corresponding to time $t=0.00233$ of the dynamics induced by $H^{4a}$ when the first regularly spaced local fractal structures resembling the sandpile identity become visible (compare Figure~\ref{Fig2}E). D\&E) Dynamics induced by $h^{2a}_{ij}=ij$, respectively $h^{2b}_{ij}=i^2-j^2$, starting at a configuration corresponding to time $t=0.5$ of the dynamics induced by $H^{4b}$ (compare Figure~\ref{Fig2}F).
	}
	\label{FigS1}
\end{figure}

\begin{figure}[htb]
	\centering
	\includegraphics[width=0.73\linewidth]{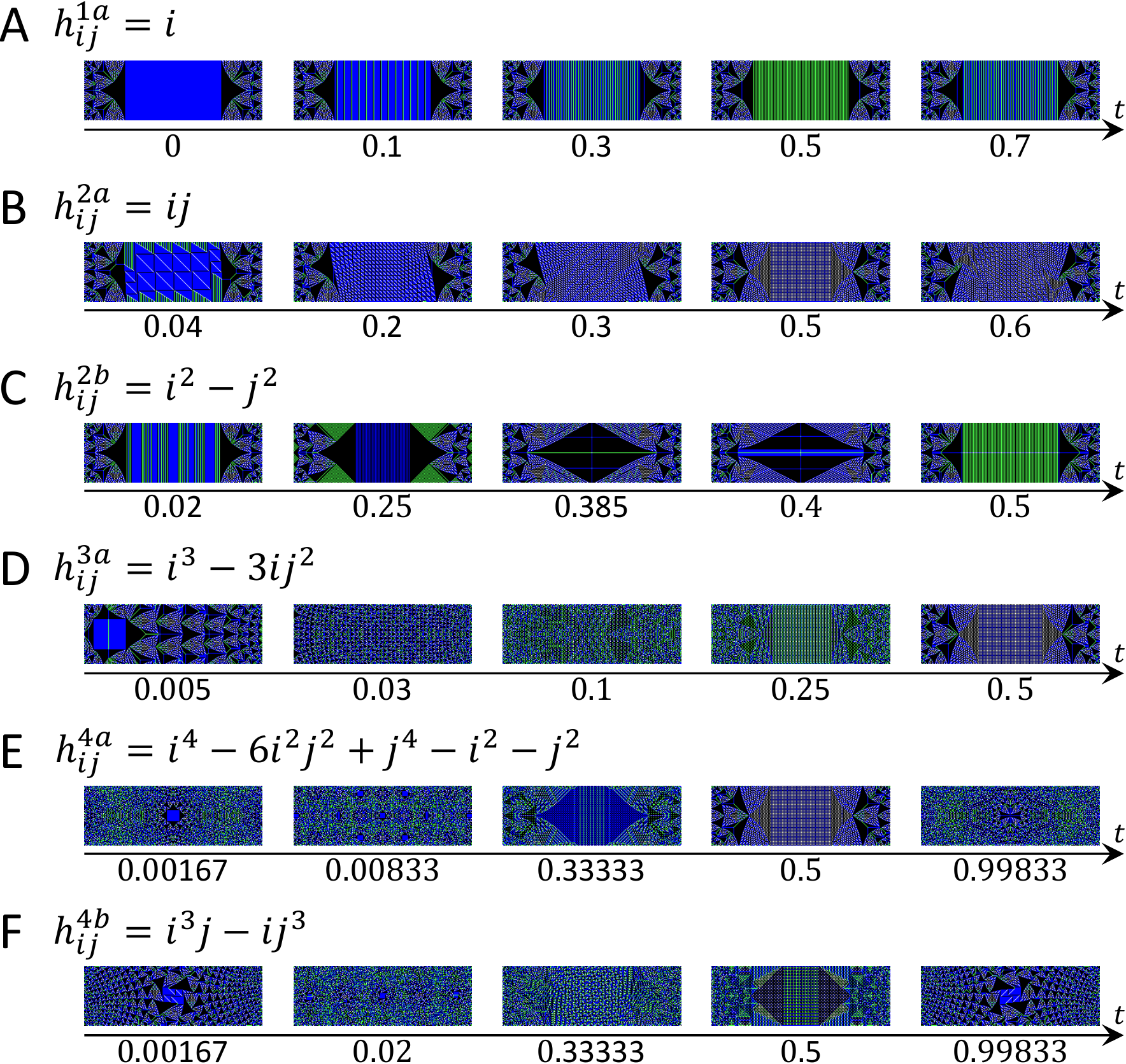}
	\caption{
	Sandpile identity dynamics on a $243\times 81$ rectangular domain. The sandpile identity dynamics were induced with the harmonic fields A) $h_{ij}^{1a}=i$, B) $h_{ij}^{2a}=ij$, C) $h_{ij}^{2b}=i^2-j^2$, D) $h_{ij}^{3a}=i^3-3ij^2$, E) $h_{ij}^{4a}=i^4-6i^2j^2+j^4-i^2-j^2$, and F) $h_{ij}^{4b}=i^3j-ij^3$.
	}
	\label{FigS2}
\end{figure}

\begin{figure}[htb]
	\centering
	\includegraphics[width=0.73\linewidth]{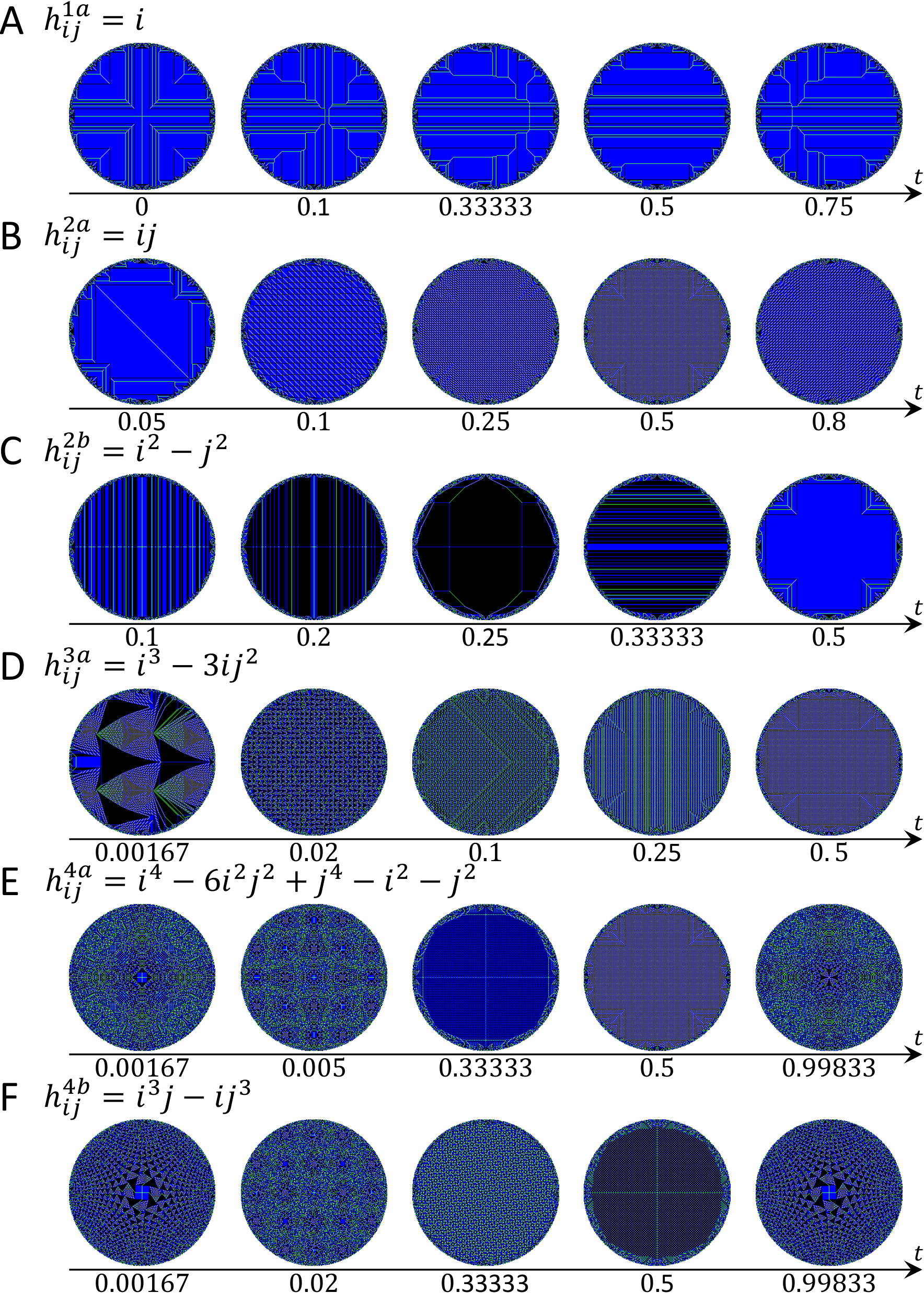}
	\caption{
	Sandpile identity dynamics on a circular domain with diameter $255$. The sandpile identity dynamics were induced with the harmonic field A) $h_{ij}^{1a}=i$, B) $h_{ij}^{2a}=ij$, C) $h_{ij}^{2b}=i^2-j^2$, D) $h_{ij}^{3a}=i^3-3ij^2$, E) $h_{ij}^{4a}=i^4-6i^2j^2+j^4-i^2-j^2$, and F) $h_{ij}^{4b}=i^3j-ij^3$.
	}
	\label{FigS3}
\end{figure}

\begin{figure}[htb]
	\centering
	\includegraphics[width=0.73\linewidth]{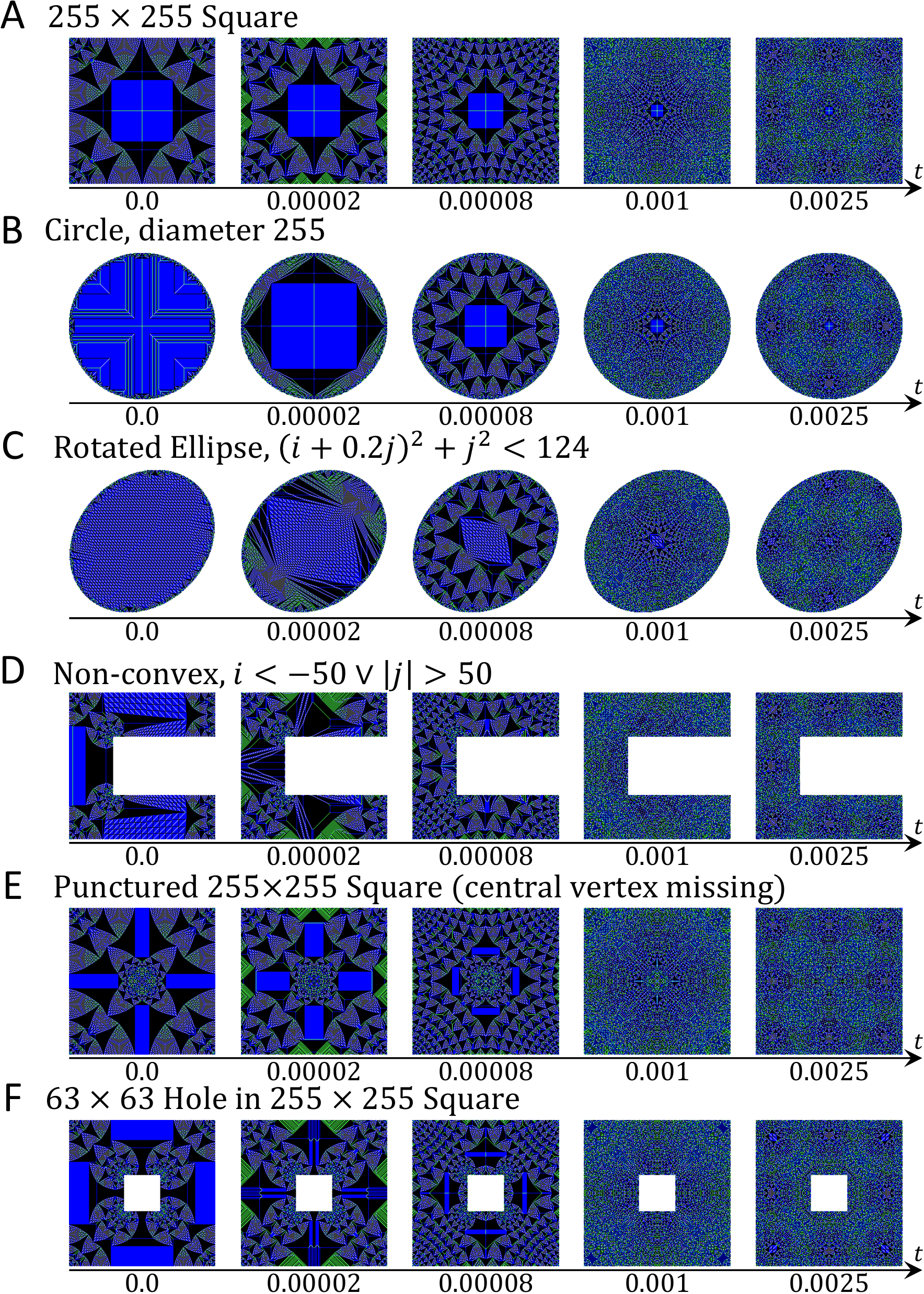}
	\caption{
	Comparison of the sandpile identity dynamics induced by $H^{4a}$ on different domains for times close to $t=0$. A) Square $255\times 255$ domain. B) Circular domain with diameter $255$. C) Rotated ellipse ($(i+0.2j)^2+j^2<124$). D) Non-convex, C-shaped domain ($i<-50\vee |j|>50$). E) Punctured $255\times 255$ domain (square domain with central vertex missing). F) $255\times 255$ square domain with $63\times 63$ hole in the center. 
	}
	\label{FigS4}
\end{figure}

\begin{figure}[htb]
	\centering
	\includegraphics[width=0.73\linewidth]{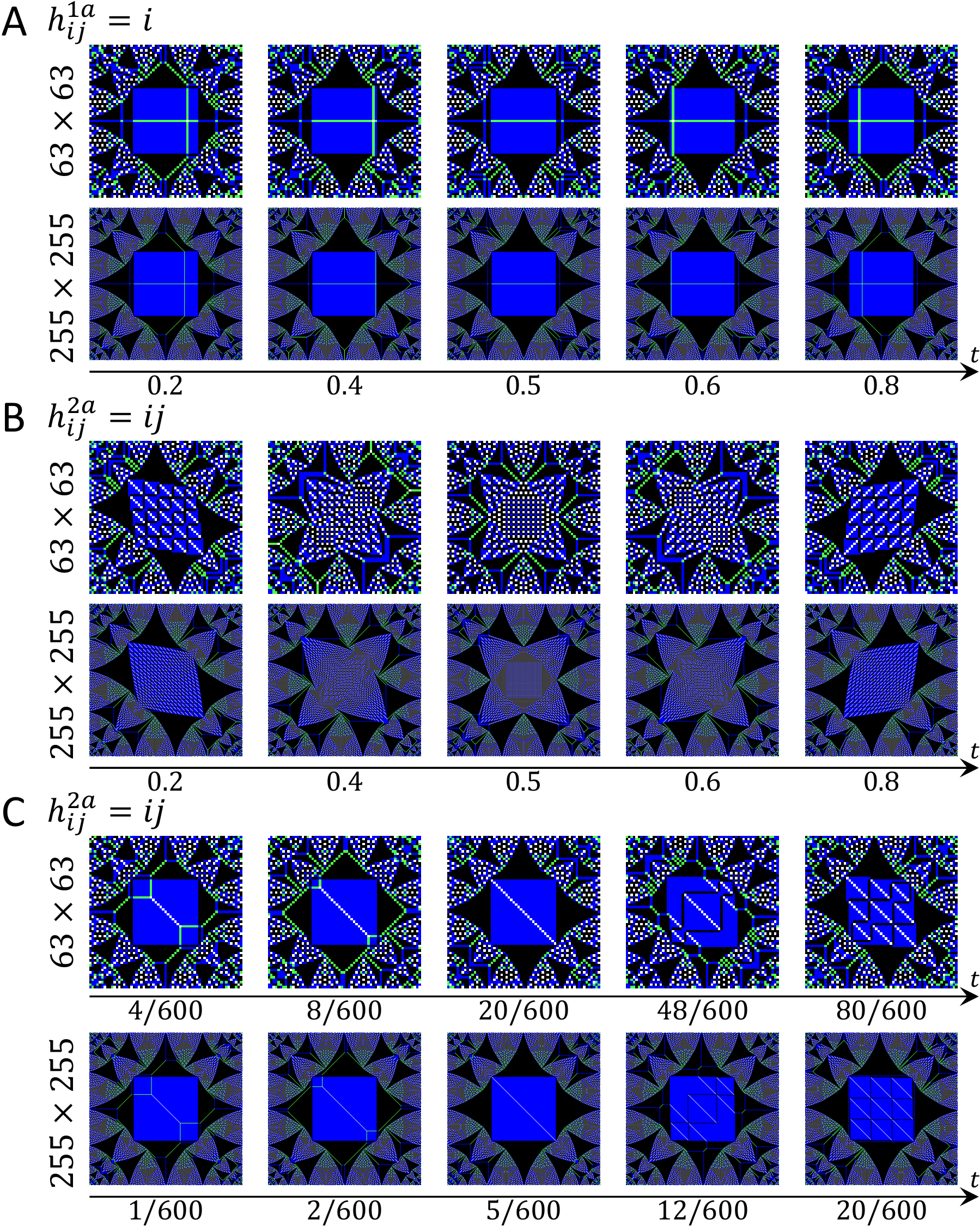}
	\caption{
	Effect of scaling the domain size (top row: $63\times 63$, bottom row $255\times 255$) on the sandpile dynamics induced by $H^{1a}$ (A) and $H^{2a}$ (B and C). 
	}
	\label{FigS5}
\end{figure}

\begin{figure}[htb]
	\centering
	\includegraphics[width=0.73\linewidth]{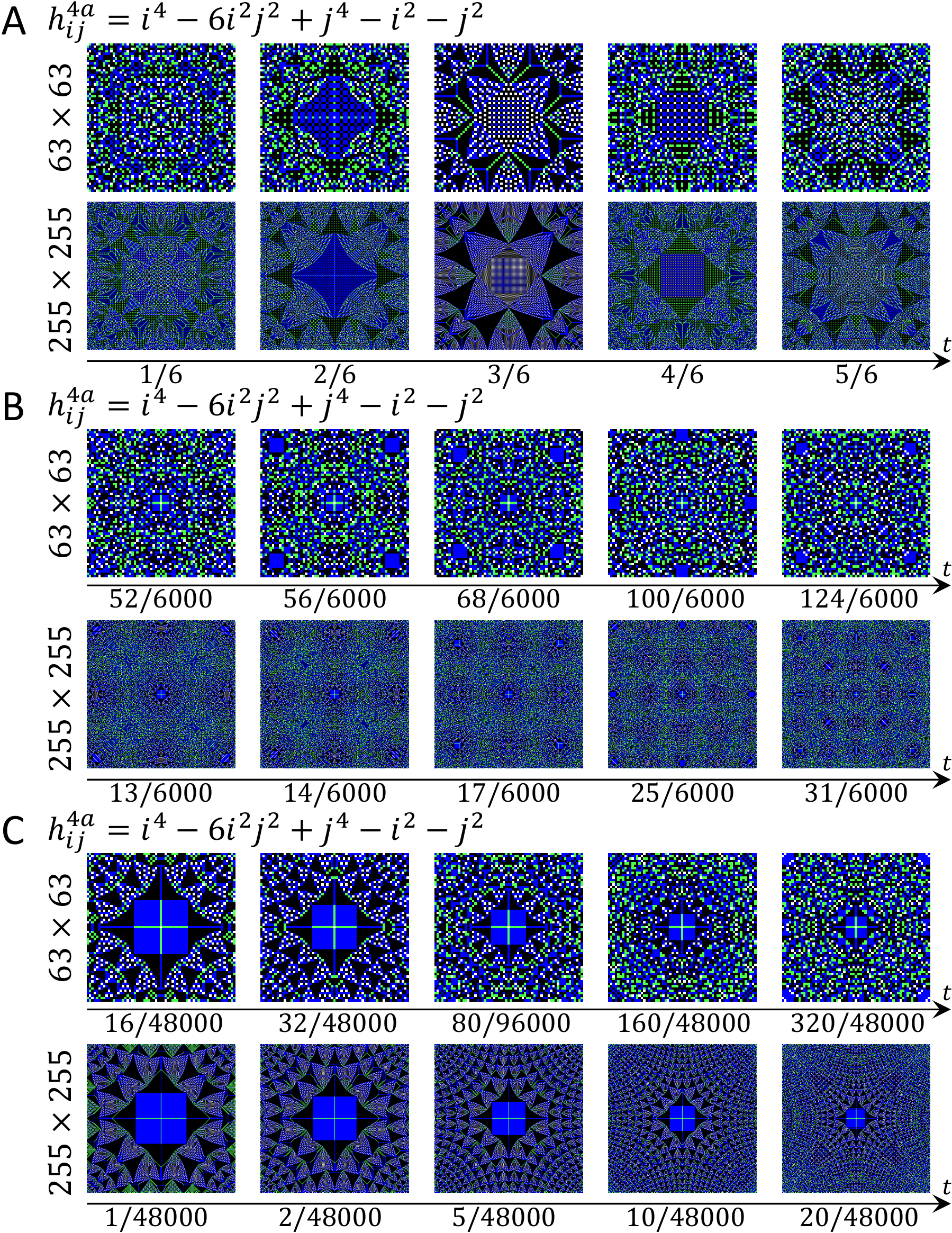}
	\caption{
	Effect of scaling the domain size (top row: $63\times 63$, bottom row: $255\times 255$) on the sandpile dynamics induced by $H^{4a}$. A) Hyper-fractals occur at the same absolute times, independently from the domain size. B) To map the positions of the regularly spaced local fractal structures resembling the sandpile identity between domains of different sizes, time has to be scaled by a factor proportional to the domain size (here by $255/63\approx 4$). C) To map patches between domains of different size, time has to be scaled by a factor proportional to the square of the domain size (here by $(255/63)^2\approx 16$).
	}
	\label{FigS6}
\end{figure}

\begin{figure}[htb]
	\centering
	\includegraphics[width=0.73\linewidth]{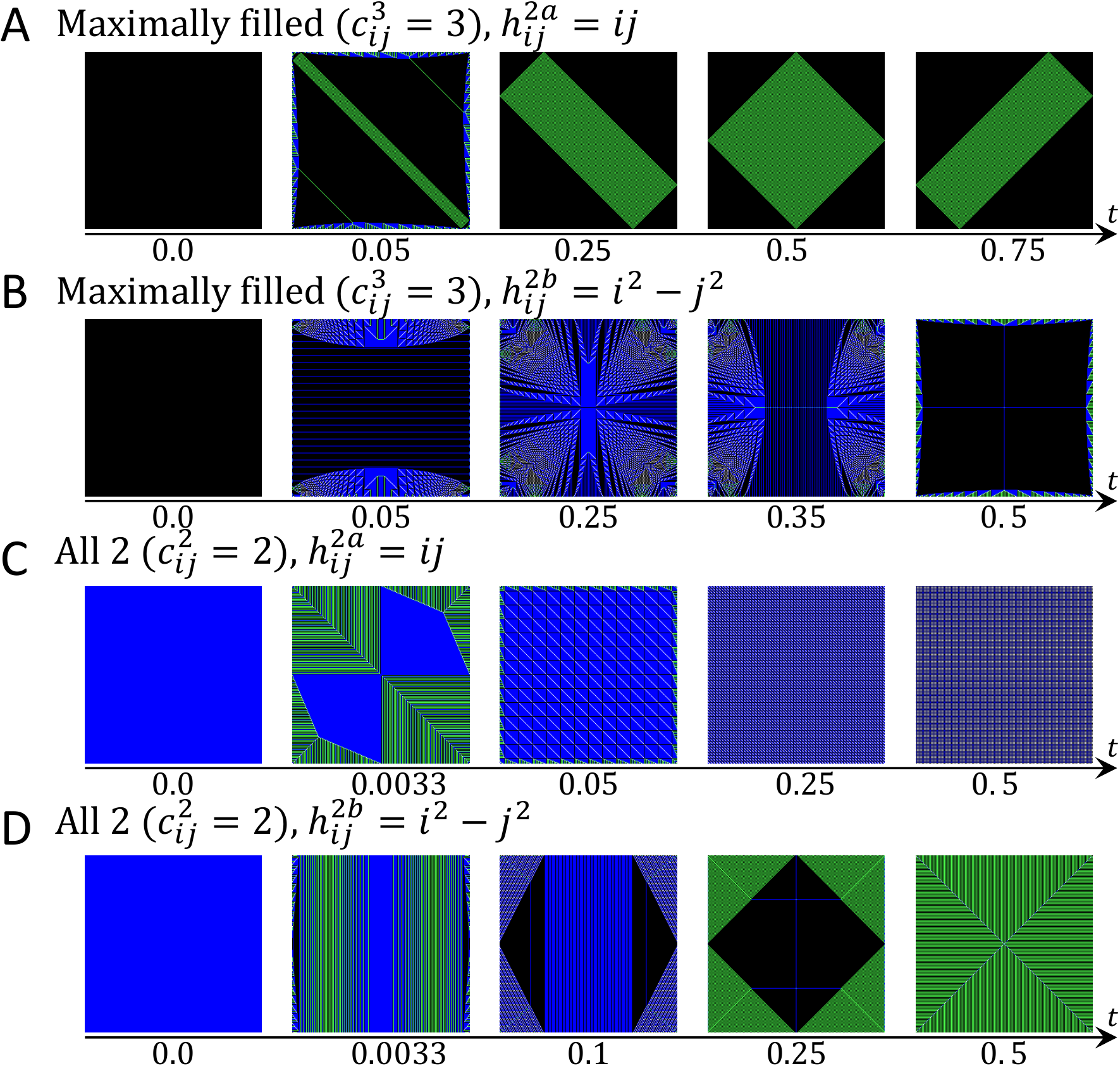}
	\caption{
	Sandpile dynamics when starting at the maximally filled configuration where each vertex carries $c_{ij}^3=3$ particles (A\&B), or the configuration where each vertex carries $c_{ij}^2=2$ particles (C\&D). For each of the two initial configurations, the dynamics induced by $h_{ij}^{2a}=ij$ (A\&C) and $h_{ij}^{2b}=i^2-j^2$ (B\&D) are shown.
	}
	\label{FigS7}
\end{figure}

\begin{figure}[htb]
	\centering
	\includegraphics[width=0.73\linewidth]{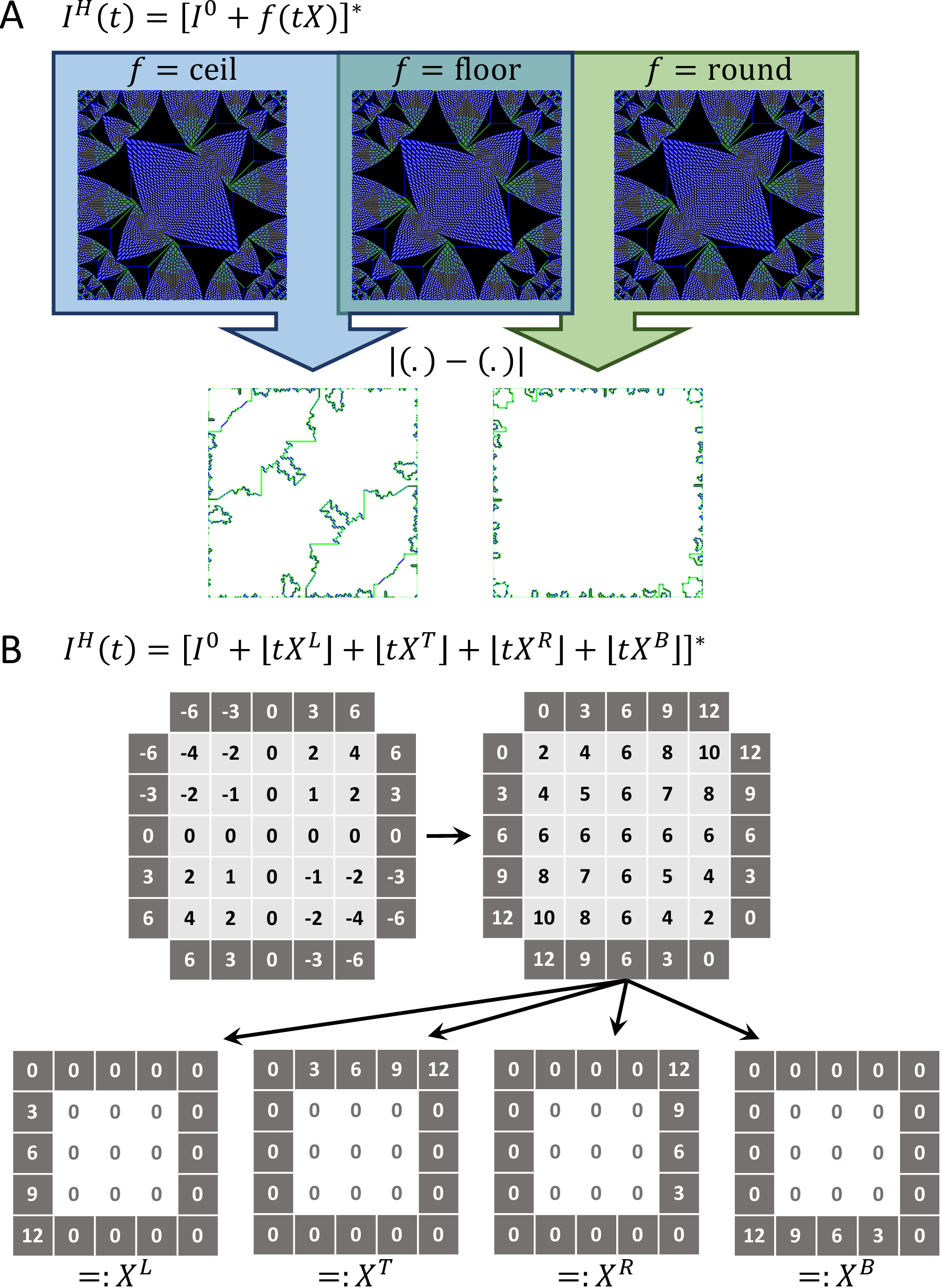}
	\caption{
	Variations of our approach to determine the sandpile identity dynamics. A) Instead of the floor function in Eq.~\ref{dynDef}, one might alternatively take the ceil or the round function. The different choices lead to nearly identical sandpile identity dynamics (here: the dynamics induced by $h_{ij}^{2a}$ at $t=1/3$). Only when directly comparing the number of particles carried by individual vertices, small differences seemingly restricted to the positions of individual tropical curves become visible. B) Instead of one potential as shown in Figure~\ref{Fig1}C, one might define four different potentials for each of the four cardinal directions, thus distinguishing from where from the outside particles are dropped onto the sandpile. This results in a slightly different dropping order onto vertices at the domain corners, while the dropping order onto vertices at the domain edges remains the same. For example, the algorithm depicted in Figure~\ref{Fig1}C results in one particle drop on the vertex at the the left bottom corner at each multiple of $t=1/24$, while the algorithm depicted here results in two particle drops at each multiple of $t=1/12$. While for most harmonics, this results in only minimally different induced sandpile identity dynamics, this algorithm has the interesting property that the constant harmonic function $h^0_{ij}=1$ induces no dynamics at all. It might also better resemble physical systems, where dynamics induced by harmonic fields of order one or higher might e.g. arise due to differences in air pressure around some object -- with the resulting forces being perpendicular to the object's surface -- while no dynamics at all arise as long as their exists no pressure gradient.
	}
	\label{FigS8}
\end{figure}

\end{document}